\RequirePackage{fix-cm}
\documentclass[smallextended]{svjour3}
\smartqed
\usepackage{amsmath,amssymb,paralist,subfigure,graphicx,cite}
\usepackage{algorithm2e}
\usepackage{color}

\def\mb{\mathbf}

\def\mc{\mathcal}

\newtheorem{thm}{Theorem}
\newtheorem{defn}{Definition}
\newtheorem{rem}{Remark}

\newtheorem{exm}{Example}

\begin{document}
	
	\title{On the Redundant Distributed Observability of Mixed Traffic Transportation Systems
	}
	\author{Mohammadreza Doostmohammadian \and Usman A. Khan \and Nader Meskin
	}
	
	\institute{ M.  Doostmohammadian \at
		Department of Mechatronics, Faculty of Mechanical Engineering, Semnan University, Semnan, Iran and Center for International Scientific Studies and Collaborations, Tehran, Iran. \\
		\email{doost@semnan.ac.ir}		
		\and		
		 U. A. Khan \at
		  Computer Science Department, Boston College, Boston, USA. \\
		 \email{usman.khan@bc.edu}
		 \and		
		 N. Meskin \at
		 Electrical Engineering Department, Qatar University, Doha, Qatar. \\
		 \email{nader.meskin@qu.edu.qa}
		 \and
		 Corresponding Author: M.  Doostmohammadian,  \email{doost@semnan.ac.ir}
		}
	
	\date{Received: date / Accepted: date}

	\maketitle
	
	\begin{abstract}
 In this paper, the problem of distributed state estimation of human-driven vehicles (HDVs) by connected autonomous vehicles (CAVs) is investigated in mixed traffic transportation systems. Toward this, a distributed observable state-space model is derived, which  paves the way for estimation and observability analysis of HDVs in mixed traffic scenarios. In this direction, first, we obtain the condition on the network topology to satisfy the distributed observability, i.e., the condition such that each HDV state is observable to every CAV via information-exchange over the network. It is shown that strong connectivity of the network, along with the proper design of the observer gain, is sufficient for this. A distributed observer is then designed by locally sharing estimates/observations of each CAV with its neighborhood. Second, in case there exist faulty sensors or unreliable observation data, we derive the condition for \textit{redundant} distributed observability as a $q$-node/link-connected network design. This redundancy is achieved by extra information-sharing over the network and implies that a certain number of faulty sensors and unreliable links can be isolated/removed without losing the observability. Simulation results are provided to illustrate the effectiveness of the proposed approach. 

	\end{abstract}

	\section{Introduction} \label{sec_intro}
 Intelligent Transportation Systems (ITS) represent a critical advancement in modern transportation, using interconnected vehicles and infrastructure to improve traffic safety and efficiency \cite{bazzan2022introduction}. The emergence of connected autonomous vehicles (CAVs) offers a promising solution for future smart cities as compared to human-driven vehicles (HDVs), CAVs are expected to have shorter reaction times and more precise control.
Although CAVs offer a promising opportunity to reduce road accidents and enhance traffic efficiency in the future, the widespread adoption of fully autonomous vehicles within the traffic network has yet to be realized. As a result, the coexistence of different vehicle types within a shared transportation system, known as mixed traffic, is expected in the coming decades. The presence of  HDVs introduces a fundamental challenge: ensuring safety and resilience despite the unpredictability of human behavior. One of the main challenges in mixed traffic scenarios is ensuring that each CAV  can accurately estimate
the state of the nearby HDVs using the available shared information in the network. To tackle this problem, this paper proposes a distributed observer design and observability analysis in mixed traffic ITS.

Traditional centralized approaches \cite{biroon2021false,ruggaber2021novel,rostami2020state,jiang2021observer} to observer design often fall short due to their vulnerability to single point of failure and scalability issues, and hence, a distributed observer framework is more appealing. In this paper,  a distributed observable state-space model framed for mixed traffic ITS is introduced that extends the observer design into a distributed framework, accommodating the dynamic states of vehicles, their sensing capabilities, and the information-sharing network. Formulating a concatenated model allows comprehensive observability and analysis across large-scale networks of CAVs. The key challenges include the design of the network topology to satisfy the distributed observability conditions and the resilience of state estimation in the presence of faulty sensors \cite{hajshirmohamadi2019distributed}, unreliable observation data \cite{tnse_attack}, or lossy communication networks \cite{icrom}.

\emph{Related Literature:} A review of sensor/infrastructure for connected vehicles in terms of sensor fusion and sensor placement problems is given in \cite{fabris2025efficient,yeong2021sensor}.
The work \cite{zhang2023intermediate} presents a general networked framework of dynamical systems for robust observer-based design. In \cite{chen2023distributed,WANG20174039}, a resilient distributed observer design based on zonotopic set membership is proposed. Such models can be used for application in ITS monitoring, however, with the assumption of local observability at every node. Adding sensor redundancy to monitor traffic flow through wireless sensor networks is considered in \cite{gagliardi2024trustworthy}.  On the other hand, the impact of the data-sharing network topology on the resilience of distributed estimation in vehicle platoons \cite{pirani2022impact}, the mixed traffic dynamics of connected vehicles \cite{ruan2022impacts}, and resilience against adversarial cyberattacks \cite{wei2023hierarchical} or false data injection attacks \cite{zhu2023distributed} are studied in the ITS literature. The effect of communication time delays on string stability as well as the minimum time headway \cite{abolfazli2023minimum}, and multi-vehicle formation tracking under actuator faults \cite{wu2023finite} are also considered. The existing literature also discusses: cooperative mesh stability \cite{qiu2022cooperative}, game-theoretic framework for improving the resilient formation of vehicle platoons \cite{pirani2020graph}, distributed formation-based tracking \cite{tase}, and stability analysis of mixed traffic flow \cite{yu2023assessment,luo2024stabilizing,zakerimanesh2024stability}. The existing works on distributed estimation include consensus filters \cite{he2021secure,qian2022consensus,battilotti2021stability} with an inner consensus loop to relax the observability assumption and network connectivity with many iterations of consensus/communication between every two steps of system dynamics or fault-tolerant estimator design \cite{ghods2025resilient,panigrahi2016fault} resilient to faults/anomalies in the sensor measurements. These filters require much faster processing and communication than sampling system dynamics, which might be practically burdensome
	in real-world ITS setup. 
What is missing in the existing literature is a general framework to model mixed traffic ITS observability by incorporating the vehicle's dynamics with the interconnection network. Such a unified framework allows for handling observability requirements via network connectivity properties and resilient design.  

\emph{Contributions:}
Through this work, we establish the foundation  for robust, scalable, and fault-tolerant observability  in mixed traffic ITS, paving the way for improved CAVs autonomy and enhanced traffic management. Our key contributions are as follows: (i) We introduce a distributed observable state-space model that integrates general HDV dynamics, CAVs' sensing capabilities, and data-exchange networks into a unified framework. Unlike traditional centralized models, our approach enables decentralized observation (and control), enhancing scalability and robustness in large-scale ITS. The proposed concatenated model effectively captures the complexities of vehicle dynamics and communication networks, providing a system-of-systems or network-of-networks representation, as described in \cite{chapman2016multiple}. (ii) We derive the specific connectivity condition on the network topology to ensure distributed observability. Our work demonstrates that strong network connectivity, coupled with a well-designed observer gain, is sufficient for each HDV state to be observable by CAVs. This improves upon existing methods that often rely on local observability conditions in the neighborhood of each CAV \cite{biroon2021false,ruggaber2021novel,rostami2020state,jiang2021observer}. Compared to the other distributed approaches \cite{modalavalasa2021review,vtc,guo2017distributed,abdelmawgoud2020distributed,yang2023state}, in this work, the network strong-connectivity is a more relaxed condition and allows for redundant network design as our next contribution. (iii) We address the issue of faulty sensors and unreliable data by introducing the concept of redundant distributed observability. Our work defines the conditions under which a $q$-node/link-connected network can tolerate and isolate a specific number of faulty sensors without losing distributed observability. This redundancy aspect is often overlooked in the existing methods, which typically do not account for sensor faults in a distributed filtering setup. (iv) Finally, the paper proposes a distributed observer design in which CAVs share estimates and observations locally with their neighbors. This localized sharing reduces communication overhead and increases system resilience compared to centralized and other existing distributed techniques that require more data-sharing and extensive communication infrastructure.

\emph{Organization of the Paper:}
Section~\ref{sec_prob} formulates the main unified distributed framework for mixed traffic ITS. Section~\ref{sec_main} presents our main results on distributed observability conditions, redundant design, and the proposed distributed observer. Section~\ref{sec_examp} provides some illustrative simulations, and Section~\ref{sec_con} concludes the paper.

\section{Problem Formulation} \label{sec_prob}
In this section, we first formulate \textit{distributed observability} of the  mixed traffic ITS from the first principles. In this formulation, we consider a networked control system (also referred to as the system-of-systems) that includes a group of CAVs to track the state of HDVs in a mixed traffic scenario. A global dynamics is derived for the network  based on the constituent general vehicle dynamics and the data-sharing network. 

Consider the following linear dynamical system as a general model describing the dynamics of a group of $N$ HDVs as part of a mixed traffic network:
\begin{eqnarray}\label{eq_sys1} 
	\mb{x}_{k+1} = A\mb{x}_k + \nu_k,
\end{eqnarray}
where~$\mb{x}_k=[\mb{x}_{1,k},\ldots,\mb{x}_{N,k}]^\top \in\mathbb{R}^{N m}$ and $\mb{x}_{i,k} \in\mathbb{R}^{m}$ is the state of the $i$th HDV at time $k$,~$A=[a_{ij}] \in \mathbb{R}^{Nm\times Nm}$ denotes the overall state-space system matrix, and~$\nu_k$ represents the noise/disturbance input\footnote{Some literature consider an input matrix $B$ associated with this random input as $B\nu_k$. For example, later in this paper, the nearly-constant-acceleration model in Example~\ref{exm_nca} considers an input matrix and random unknown input variable that gives the system model as $\mb{x}_{k+1} = A\mb{x}_k + B\nu_k$}. Each diagonal block of the system matrix $A$, denoted by $\widetilde{A}_i$ with $i\in \{1,\ldots,N\}$, represents the dynamics associated with the $i$th HDV\footnote{For homogeneous ITS the dynamics of all HDVs are assumed to be the same and $\widetilde{A}_i=\widetilde{A}_j$ for all $i,j \in \{1,\ldots,N\}$}. 
The embedded sensors at $n$ CAVs take  observations of the state of  $N$ HDVs, represented as
\begin{eqnarray} \label{eq_H_i}
	\mb{y}_{i,k} = C_i\mb{x}_k + {\mu}_{i,k},~~~i=1,\dots,n
\end{eqnarray}
where $\mb{y}_{i,k}\in\mathbb{R}^{l_i}, ~C_i\in\mathbb{R}^{l_i\times Nm}$, and~${\mu}_{i,k}$  denote the measurement vector, local output matrix, and noise at the $i$th CAV, respectively. 
The global observation vector of the entire CAVs is then defined as
\begin{eqnarray} \label{eq_H}
	\mb{y}_k = C\mb{x}_k + \mu_k,
\end{eqnarray}
with~$\mb{y}_k \in\mathbb{R}^{L},~C=[C_1^\top,C_2^\top,\ldots,C^\top_n]^\top$ as the \textit{global} measurement vector and output matrix, respectively,
with~$L=l_1+\ldots+l_n$ and~$\mu_k$ as the global noise. 
The problem of estimating the state of the HDVs can be fundamentally considered in two main scenarios:
\begin{itemize}
	\item \textit{Centralized setup:} In this setup, the sensor measurements at different CAVs are transmitted to a central coordinator for state estimation and filtering purposes. Given the noisy sensor measurements, using a centralized observer scheme, the central coordinator is able to estimate the entire system states at all HDVs if and only if the pair $(A,C)$ (or $(A,C_i)$) is observable \cite{bay}. See \cite{biroon2021false,ruggaber2021novel,rostami2020state,jiang2021observer} as examples of centralized observer design.
	\item \textit{Distributed setup:} The processing units are embedded at each CAV such that they can interact over a communication network with a given graph topology. Each vehicle then estimates the  state, ${\mb{x}}_{k}$, given the local sensor measurements and/or
	\textit{predictions} (or a-priori estimates) from the neighbouring CAVs. Given that the entire system is not observable to any single CAV, we formally derive the condition for \textit{distributed observability}. This implies that under such conditions each CAV is able to estimate the state of all HDVs over time only by local data-exchange in its neighbourhood. See \cite{modalavalasa2021review,vtc,abdelmawgoud2020distributed,guo2017distributed,he2020distributed,yang2023state} (and references therein) for examples of distributed observer algorithms for general applications.
\end{itemize}
{To establish the foundations for the distributed ITS observability problem, first consider the graph topology $\mc{G}_W=(\mc{V}_W,\mc{E}_W)$ as the interaction network among the CAVs. A link $(j,i) \in \mc{E}_W$ implies that CAV $j$ sends information to CAV $i$. As we see later in the paper, we need the information-sharing (directed or undirected) network of CAVs to be strongly-connected. Let the matrix $W=[w_{ij}] \in \mathbb{R}^{n\times n}$ be the adjacency weight matrix of $\mc{G}_W$ which needs to be row-stochastic\footnote{Matrix $W$ is called row-stochastic if we have $\sum_{j=1}^n w_{ij}=1$ for every $i \in \{1,\dots,n\}$. } and $\mc{N}_i$ denotes the neighbours of vehicle~$i$ over this network. A straightforward way to assign stochastic weights to the links associated with each CAV is to consider $w_{ij}=\frac{1}{\mc{N}_i}$. However, in a more general case, the weights at different links may differ while satisfying the row-stochastic property. This row-stochastic property simply implies that each CAV $i \in \mc{V}_W$ averages the information sent from the neighbouring CAVs in $\mc{N}_i \in \mc{V}_W$, which leads to an \textit{agreement} or \textit{consensus} state over time.  }
By exchange of information over $\mc{G}_W$, the $i$th CAV  now estimates the global state of HDVs, $\mb{x}_k$, with its own possible observation, $\mb{y}_{i,k}$, and information received from its neighbours $\mb{y}_{j,k}$, $j\in\mc{N}_i$. At the $i$th CAV,  thus, the problem is to estimate the state of Eq.~\eqref{eq_sys1} associated with the HDVs from the following outputs:
\begin{eqnarray} \label{eq_sys3}
	\mb{y}_{j,k} = C_j\mb{x}_k + \mu_{j,k},\qquad j\in\mc{N}_i,i \in \mb{I}_1^n,
\end{eqnarray}
with $\mb{I}_1^n$ denoting the set $\{1,\dots,n\}$ or via its equivalent redefined model:
\begin{eqnarray}\label{eq_zk}
	\mb{z}_{i,k} = \sum_{j\in\mc{N}_i}C_j^\top C_j\mb{x}_k + \sum_{j\in\mc{N}_i}C_j^\top\mu_{j,k},
\end{eqnarray}
assuming that $\mc{N}_i$ includes node $i$ itself. Eq.~\eqref{eq_zk} captures the shared measurements in the neighborhood of $i$th CAV.  Consider the case that the shared information only includes the measurements $\mb{y}_{j,k},j\in\mc{N}_i$; it immediately follows that this estimation problem can be only realized if and only if the pair $(A,\sum_{j \in \mc{N}_i} C_j^\top C_j)$ is observable \cite{sauter:09}. Thus, the decentralized observability \emph{in this case} can be defined based on observability of the pair $(I\otimes A, D_C)$ with the following block-diagonal observation matrix:
\begin{eqnarray} \label{D_H}
	D_C =
	\left(
	\begin{array}{ccc}
		\sum_{j\in\mc{N}_1}C_j^\top C_j&&\\
		&\ddots&\\
		&&\sum_{j\in\mc{N}_n}C_j^\top C_j
	\end{array}
	\right).
\end{eqnarray}
and $\otimes$ as the Kronecker product. This argument represents a trivial semi-centralized solution for observability, where the CAV $i$ needs to be densely connected in its neighbourhood \cite{sauter:09,das2016consensus} such that all its limited sensor measurements are widely shared over the network $\mc{G}_W$. However, intuitively, it is expected to have a more relaxed setup for distributed observability over $\mc{G}_W$. In the following discussions, it is shown  that distributed observability, in its most general case, \emph{does not} require that the state of HDVs be \textit{directly measured} and observable in its neighbourhood. This follows adding a step of \emph{a-priori estimate} or \emph{prediction} fusion as discussed below.

To generalize the distributed estimation problem, assume that the pair $(A,\sum_{j \in \mc{N}_i} C_j^\top C_j)$ is not necessarily observable in the $i$th CAV, $i \in \mb{I}_1^n$. In what follows, it is shown that distributed observability does not require this condition. To formulate the most general setup, consider~$\widehat{\mb{x}}^i_{k|k}\in\mathbb{R}^{N m}$ as the estimate of $\mb{x}_k\in\mathbb{R}^{N m}$ by using all the observations accessible to the $i$th CAV  (up to time $k$) over the communication network, $\mc{G}_W$. Let the \emph{global} estimate of the HDVs' states be represented by column concatenation of the local estimates as 
\begin{eqnarray}
	\widehat{\mb{x}}_{k|k}:=
	\left(
	\begin{array}{c}
		\widehat{\mb{x}}^1_{k|k}\\
		\widehat{\mb{x}}^2_{k|k}\\
		\vdots \\
		\widehat{\mb{x}}^n_{k|k}
	\end{array}
	\right)\in\mathbb{R}^{nN m}
\end{eqnarray}
Now, the above represents the estimate of the global state representing the ITS as a networked system (or system-of-systems) given as 
\begin{eqnarray}\label{eq_xx}
	\underline{\mb{x}}_{k}:=
	\left(
	\begin{array}{c}
		\mb{x}_k\\
		\mb{x}_k\\
		\vdots \\
		\mb{x}_k
	\end{array}
	\right) = \mb{1}_n \otimes \mb{x}_k.
\end{eqnarray}
By combining \eqref{eq_sys1} and ~\eqref{eq_xx},  the dynamics equation corresponds to $\underline{\mb{x}}_{k}$ is given as follows:
\begin{eqnarray}
	\underline{\mb{x}}_{k+1} &=& \mb{1}_n\otimes \mb{x}_{k+1}\nonumber \\
	&=& \mb{1}_n\otimes (A\mb{x}_{k}+\nu_k)\nonumber \\
	&=& \mb{1}_n\otimes A\mb{x}_{k}+\mb{1}_n\otimes\nu_k\nonumber \\
	&=& {(W\otimes A)}\underline{\mb{x}}_{k}+{\mb{1}_n\otimes\nu_k}\nonumber\\
	&=& Q \underline{\mb{x}}_{k} + \underline{\nu}_{k},\label{eq_sys5}
\end{eqnarray}
where $Q=W\otimes A$, $\underline{\nu}_{k}={\mb{1}_n\otimes\nu_k}$  and the last equation follows if and only if the adjacency matrix $W$ is row-stochastic following the consensus nature of the design  \cite{olfatisaberfaxmurray07}. Recall that, following the row-stochastic property, we have $W\mb{1}_n= \mb{1}_n$ and, therefore, one can write $\mb{1}_n\otimes A\mb{x}_{k} = W\mb{1}_n\otimes A\mb{x}_{k} = (W\otimes A)(\mb{1}_n\otimes \mb{x}_{k})=  Q \underline{\mb{x}}_{k} $. Thus, the ITS dynamics can be associated with $Q=W\otimes A\in \mc{Q}$ as its system matrix, and we have
\begin{eqnarray}
	\mc{Q} = \{Q~|~Q=(W\otimes A) ~ \mbox{and}~ W \mbox{ is stochastic.}\}
\end{eqnarray}
where the set $\mc{Q}$ captures the structure of the general dynamics. From the perspective of structured systems theory \cite{rein_book,ramos2022overview}, $\mc{Q}$ represents a \emph{class} of system matrices that follow a fixed, time-invariant structure, while its non-zero entries may vary over time. In general, by choosing any matrix in $\mc{Q}$, the dynamics by \eqref{eq_sys5} remains a valid representation of the concatenated ITS model. This is due to the fact that any matrix in $\mc{Q}$ can be structurally represented as the Kronecker product of the communication adjacency matrix $W$ and the vehicle dynamics $A$, where the nonzero entries may change depending on the given dynamics of HDV.

The above arguments imply that distributed estimation is equivalent to the centralized estimation of the following system:
\begin{eqnarray}
	\underline{\mb{x}}_{k+1} &=& Q \underline{\mb{x}}_{k} + \underline{\nu}_{k+1},\qquad Q\in\mc{Q}, 
\end{eqnarray}
and with output variable $\mb{z}_k$ following from the column concatenation of $\mb{z}_{i,k}$ in \eqref{eq_zk} as
\begin{eqnarray}
	\mb{z}_k := & D_C\underline{\mb{x}}_{k} + \widetilde{\mu}_k,
\end{eqnarray}
where $\widetilde{\mu}_k$ is the concatenation of $\widetilde{\mu}_{i,k} = \sum_{j\in \mc{N}_i}C_j^\top\mu_{j,k}$. Notice that both $Q$ and $D_C$ matrices depend on the communication network $\mc{G}_W$ among the CAVs.  A better illustration of this distributed problem setup is given in Fig.~\ref{fig_platoon}.
\begin{figure} 
	\centering
	\includegraphics[width=3.5in]{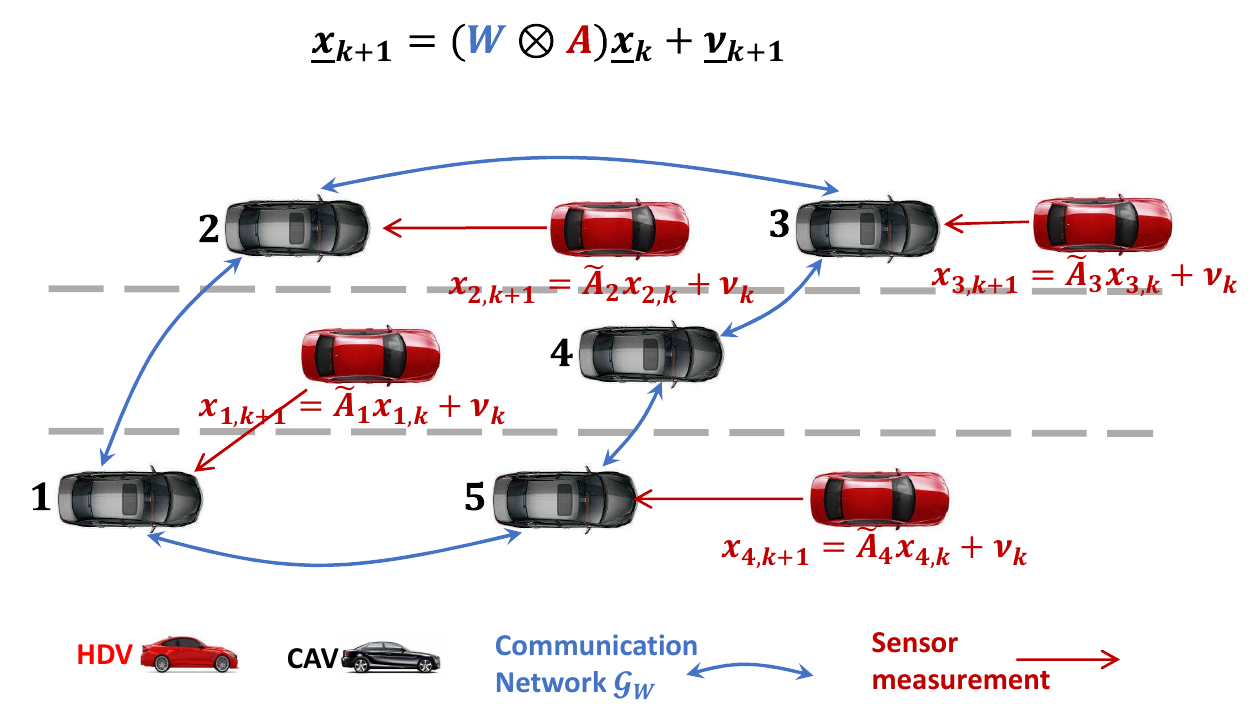}
	\caption{This figure illustrates the distributed observer design formulation in this paper. The system state to be estimated includes the state $\mb{x}_k$ of the HDVs. Each CAV aims to observe the entire state of the HDVs. The global state estimate dynamics can be represented as the Kronecker product of the vehicle's system matrix $A$ and the adjacency matrix $W$ of the communication network $\mc{G}_W$. Some nearby CAVs take measurements of the state of the HDV. Then, the observability of the global HDVs' dynamics depends on both the local system matrix $A$ and the graph topology $\mc{G}_W$ along with the local measurement matrix $C$ (or $D_C$).
	} \label{fig_platoon}
\end{figure}
All these results are summarized in the following remark. 
\begin{defn} \label{1}
	The \textit{distributed observability} of the ITS modelled by a network of $n$ CAVs tracking the state of $N$ HDVs with dynamics~\eqref{eq_sys1} and outputs~\eqref{eq_H} is defined as observability of the pair
	\begin{equation} \label{eq_dist_obsrv}
		(W \otimes A, D_C),
	\end{equation}
	with $W$ as the adjacency matrix of the communication network of CAVs, $A$ as the system matrix of HDVs, and $D_C$ representing the shared observations as defined in~\eqref{D_H}. 
\end{defn}
An example of such a distributed estimation setup for general dynamical systems is given in \cite{TNSE19,acc13}. The other concern is losing sensor data due to faults/attacks and its impact on the ITS observability. It is likely that the sensor measurements at some CAVs are faulty and not trustworthy.  Therefore, the faulty sensors need to be isolated to prevent the cascade of error over the entire transport network\footnote{In terms of faulty sensor measurements, we assume that (at least) one trustworthy sensor output from each HDV is available to the group of CAVs.}. This is done by distributed observer-based fault-detection and isolation (FDI) strategies \cite{tnse_attack,hajshirmohamadi2019distributed,zaeri2023distributed}. Isolating the faulty sensors and removing their observation data, on the other hand, may violate $(W \otimes A, D_C)$-observability. The remedy in this work is to add redundancy in the communication network of CAVs, $\mc{G}_W$, such that by sharing more data, $(W \otimes A, D_C)$-observability is recovered. Recall that this redundancy in the network improves both prediction-fusion via $W \otimes A$ term and observation-fusion via the term $D_C$. 

Given these arguments on distributed observability, the following problems are addressed in this paper:
\begin{itemize}
	\item \textbf{Problem 1:} Considering HDVs' dynamics \eqref{eq_sys1} and observations \eqref{eq_H_i}, what are the required conditions on the network graph $\mc{G}_W$ (i.e., the structure of $W$ matrix) such that  the pair $(W \otimes A, D_C)$ is observable. Given the obtained condition, how one can design distributed observers in each CAV such that it can locally (with no need for a centralized coordinator) track the entire state of the HDVs?
	\item \textbf{Problem 2:} What are the extra conditions on $\mc{G}_W$ to have \textit{redundant} distributed observability resilient to sensor/link failures? In other words, how one can add redundancy to the graph $\mc{G}_W$ such that $(W \otimes A, D_C)$-observability holds under removal/failure of certain sensor data or communication links?
\end{itemize}  

\section{Main Results} \label{sec_main}
This section provides our main results on structural distributed observability via Kronecker-based system modeling and its resilience via $q$-node/link-connectivity. Section \ref{secsub_main} gives the structural network design via Kronecker network product to satisfy distributed observability. Section~\ref{secsub_redund} adds resilience via $q$-node/link-connectivity to network design. Section~\ref{subsec_obsrv} formulates the main distributed observer and relates the error dynamics stability to the Kronecker product design.     

\subsection{Network Design for Distributed Observability} \label{secsub_main}
In this section, after discussing centralized $(A,C)$-observability,  the conditions on the topology of the communication network of CAVs, $\mc{G}_W$, are derived to satisfy $(W \otimes A, D_C)$-observability. The proposed approach follows structured systems theory \cite{rein_book,woude:03} and is based on the zero-nonzero structure of the system matrix $A$ and adjacency weight matrix $W$. In other words, the results are irrespective of the numerical values of entries in $A$ and $W$ and hold for \textit{almost all} choices of their nonzero entries \cite{woude:03}. In this direction, define the system digraph $\mc{G}_A=\{\mc{V}_A,\mc{E}_A\}$ associated with the system matrix $A$ in \eqref{eq_sys1} modelling the HDVs' dynamics. The set of nodes in $\mc{V}_A$ represents the states of HDVs and the set of links $\mc{E}_A$ represents the nonzero entries $A_{ij}\neq 0$. Given a full-rank system matrix $A$, its system graph $\mc{G}_A$ is cyclic, i.e., it contains a family of disjoint cycles in $\mc{E}_A$ spanning all the nodes $\mc{V}_A$. This generally holds for systems where all the diagonal entries are non-zero, representing self-loops at all state nodes \cite{TNSE19}. These systems are typically referred to as self-damped dynamical systems. Further, define an output of the system graph $\mc{G}_A$ as the sensor measurement/observation of a state node in $\mc{V}_A$. 
Given these notions, one can redefine the observability of system matrix $A$ and output matrix $C$ on its associated system graph $\mc{G}_A$. In other words, instead of algebraic techniques for observability assessment, such as Grammian matrix, graph theoretic techniques are used to check the system observability \cite{woude:03}. 
Recall that these graph-theoretic methods are based on zero-nonzero structure of matrices $A,C$ and therefore are called structural observability. The following gives the main theorem on the structural observability of a given system graph.

\begin{thm} \label{thm_cent}
	A cyclic system graph $\mc{G}_A$ is observable (in the structural sense) if and only if it contains a directed path (or sequence of connected nodes) from every node $i$ to an output measurement.
\end{thm}
\begin{proof}
	Ref. \cite{Liu-nature} presents the graph-theoretic proof for (structural) controllability, which can be easily extended to the dual concept of observability.
\end{proof}

Following the above theorem, the notion of observability is closely related to the following definitions in graph $\mc{G}_A$.
Define a strongly-connected-component (SCC), denoted by $\mc{S}=\{\mc{V}_S,\mc{E}_S\}$, as a component (or subgraph) in which there exists a directed path (a set of subsequent directed links) in $\mc{E}_S$ from every node $i \in \mc{V}_S$ to every other node $j \in \mc{V}_S$. A graph is called strongly-connected (SC) if all of its nodes make one giant SCC. It is known that the system matrix associated with an SC graph is irreducible\footnote{A matrix is called irreducible if it cannot be transformed into block upper/lower-triangular form by simultaneous row/column permutations.} \cite{rein_book}. If the system graph $\mc{G}_A$ is non-SC, its associated system matrix $A$ is reducible. If $\mc{G}_A$ is non-SC, it contains one or more SCCs with no outgoing link to nodes in other SCCs; such a component is called a \textit{parent} SCC, denoted by $\mc{S}^p$, and otherwise the component is a \textit{child} SCC, denoted by $\mc{S}^c$. A non-SC system graph $\mc{G}_A$ can be decomposed into \textit{disjoint} (parent and child) SCCs with their partial order, denoted by $\prec$, determined by the depth-first-search (DFS) algorithm in polynomial-order complexity \cite{algorithm}. For example, $\mc{S}_i \prec \mc{S}_j$ means that one or more directed links (or paths) from nodes in~$\mc{S}_i$ to nodes in~$\mc{S}_j$ exist. The following theorem gives the necessary and sufficient condition for the observability of a system graph $\mc{G}_A$ based on its SCC decomposition. {These are illustrated later in Example~\ref{exm_nca} and Fig.~\ref{fig_nca}.
	
	\begin{thm} \label{thm_scc}
		Given a cyclic graph~$\mc{G}_A$, the necessary and sufficient condition for structural observability of~$\mc{G}_A$ is to have an observation/output from (at least) one state node in every parent SCC. 
	\end{thm}
	\begin{proof}
		\textbf{Sufficiency}: Assume that we have an output of a state node in every parent~SCC. Therefore, all nodes in that parent SCC $\mc{S}^p_i$ are output connected. From the definition, every child SCC has an outgoing path to one or more parent SCCs. This implies that for any child SCC $\mc{S}^c_j \prec \mc{S}^p_i$ the output connectivity is also satisfied based on the definition. Therefore, the condition in Theorem~\ref{thm_cent} holds for all state nodes and structural observability follows.
		
		\textbf{Necessity}: Assume by contradiction that there is no state observation/output from one parent SCC, say~$\mc{S}^p_j$. This implies that the state nodes in~$\mc{S}^p_j$ are not output-connected. This violates the observability condition in Theorem~\ref{thm_cent}, and the proof follows.
	\end{proof}
	Note that the results of the theorem hold in case we have parent state \textit{nodes} instead of parent \textit{SCCs}. The following example illustrates this.
	
	\begin{exm} \label{exm_nca}
		Consider the vehicle system dynamics as nearly-constant-acceleration (NCA) or nearly-constant-velocity (NCV) model in \cite{bar2004estimation}. In the 2D case, the NCA dynamics for an HDV is in the form $\mb{x}_{k+1} = A\mb{x}_k + B\nu_k$, where the system matrix $A$ and the input matrix $B$ are modelled as 
		\begin{align}  \label{eq_nca}
			A = \left(
			\begin{array}{cccccc}
				1 & 0 & 0 & 0 & 0 & 0\\
				0 & 1 & 0 & 0 & 0 & 0  \\
				T & 0 & 1 & 0 & 0 & 0 \\
				0 & T & 0 & 1 & 0 & 0 \\
				\frac{T^2}{2} & 0 & T & 0 & 1 & 0 \\
				0 & \frac{T^2}{2} & 0 & T & 0 & 1 \\  
			\end{array} \right),~B = \left(
			\begin{array}{cc}
				\frac{T^2}{2} & 0 \\
				0 & \frac{T^2}{2} \\
				T & 0 \\
				0 & T \\
				1 & 0 \\
				0 & 1 \\ 
			\end{array} \right) 
		\end{align}
		with $T$ as the time constant and states 
		\begin{align}
			\mb{x} = \left(
			\begin{array}{c}
				\ddot{p}_x \\
				\ddot{p}_y \\
				\dot{p}_x \\
				\dot{p}_y \\
				p_x\\
				p_y  \\  
			\end{array} \right)
		\end{align}
		with $p_x,p_y$ as positions, $\dot{p}_x ,\dot{p}_y $ as velocities, and $\ddot{p}_x ,\ddot{p}_y $ as accelerations in $x,y$ directions. In this case, Theorem~\ref{thm_scc} implies that we need output of the parent state nodes $p_x$ and $p_y$ as position states and the other states can be inferred by observing $p_x$ and $p_y$. This is better illustrated in Fig.~\ref{fig_nca}.
		\begin{figure}
			\centering
			{\includegraphics[width=3in]{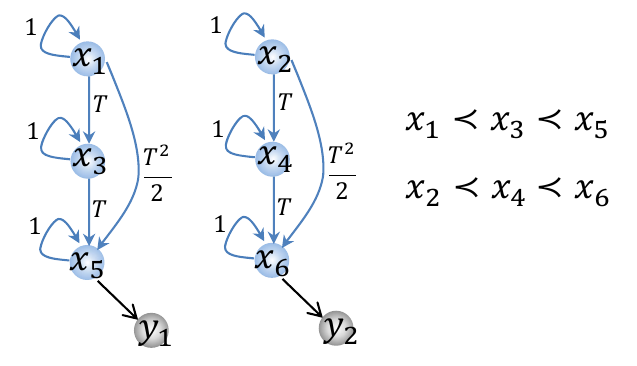}}
			\caption{ {This figure shows the system graph $\mc{G}_A$ associated with the NCA system matrix in Eq.~\eqref{eq_nca}. The partial order of SCCs (here as self-cycles) is given at the right of the figure. The position states $\mb{x}_5$ and $\mb{x}_6$ represent the parent SCCs (or parent nodes). Outputs $\mb{y}_1$ and $\mb{y}_2$ ensure structural observability of the system graph according to Theorem~\ref{thm_scc}.}}
			\label{fig_nca}
		\end{figure}
	\end{exm}
	
	Next, considering the communication network $\mc{G}_W$, we derive the graph-theoretic condition for distributed observability.
	
	\begin{thm} \label{thm_main2}
		Let the cyclic system graph $\mc{G}_A$ and the output matrix $C$ satisfy the structured observability condition in Theorem~\ref{thm_cent}. If the graph $\mc{G}_W$ is strongly connected,  then the pair $(W \otimes A, D_C)$ is observable.
	\end{thm}
	\begin{proof}
		The proof methodology follows from \cite[chapter~1]{rein_book}. Let the composite Kronecker network~$\mc{G}_W \otimes \mc{G}_A$ represent the structure of $(W \otimes A)$ matrix. First, recall that for an SC $\mc{G}_W$ network, its $W$ matrix is irreducible and the diagonal entries of $W$ matrix are all nonzero since in the distributed setup, each CAV uses its own information (i.e., its own prediction of the HDV states) for data fusion. 
		
		Therefore, the main diagonal blocks of $(W \otimes A)$ are the same as system matrix~$A$. This follows the definition of the Kronecker product of matrices as multiplying a nonzero scalar $w_{ij}$ in~$A$ does not change its structure. Let consider the system graph $\mc{G}_A$ with child and parent SCCs~$\mc{S}^c_i,\mc{S}^p_j$ for which~$\mc{S}^c_i \prec \mc{S}^p_j$. Denote the irreducible blocks associated with these SCCs as $A_{ii}$ and $A_{jj}$, respectively. Then, as illustrated in Fig.~\ref{fig_matrix2}, $\mc{S}^c_i \prec \mc{S}^p_j$ implies a sequence of irreducible blocks from $A_{ii}$ to $A_{jj}$, shown by grey arrows. 
		Then, the non-diagonal blocks of~$(W \otimes A)$ are mapped based on the irreducible structure of $W$.
		\begin{figure}
			\centering
			{\includegraphics[width=2.75in]{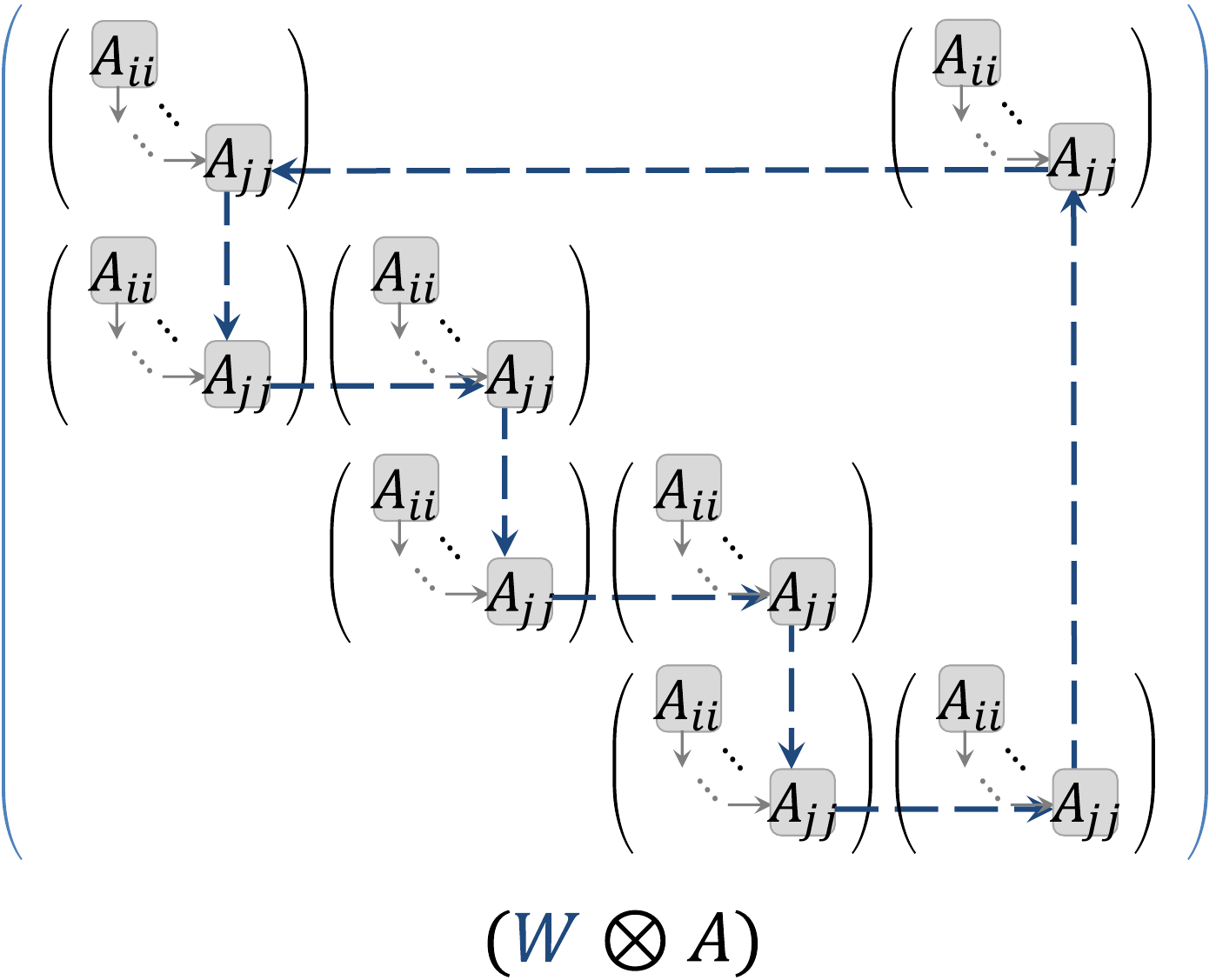}}
			\caption{ This figure presents the structure of~$W \otimes A$ matrix. Assume that every reducible block of system matrix~$A$ contains a smaller irreducible block~$A_{ii}$ (representing a child SCC in~$\mc{G}_A$) which forms a path to another irreducible block~$A_{jj}$ (representing a parent SCC in~$\mc{G}_A$). From the definition of $W \otimes A$, the blocks of $A$ matrices follow the irreducible structure of the adjacency matrix~$W$ which represents the SC communication network $\mc{G}_W$. Following the dashed blue arrows, one can find a giant irreducible block formed by $A_{jj}$s, the irreducible block obtained after a suitable permutation that collects all blocks connected by parent-child SCC paths. This giant block forms a giant parent SCC in the network $\mc{G}_W \otimes \mc{G}_A$.} 
			\label{fig_matrix2}
		\end{figure}
		One can see a sequence of non-diagonal blocks in~$(W \otimes A)$ sharing no hyper-rows and no hyper-columns\footnote{Define a hyper-row (or hyper-column) as an $Nm$ by $nNm$ row (or column) matrix formed by $A$ matrices in~$(W \otimes A)$ matrix.}. The irreducible blocks associated with parent SCCs~$\mc{S}^p_j$ (in every system matrix~$A$) form an SC path that follows the (irreducible) structure of~$W$, shown by dashed blue arrows in Fig.~\ref{fig_matrix2}. This implies a giant irreducible block in~$(W \otimes A)$ that can be obtained by row/column permutation and forms a giant SCC in~$\mc{G}_W \otimes \mc{G}_A$. This giant SCC has no outgoing links and, therefore, is a parent SCC in~$\mc{G}_W \otimes \mc{G}_A$. The irreducible structure can be generalized for any SCC in~$\mc{G}_A$, i.e., for every (child and parent) SCC in~$\mc{G}_A$, one can find a giant SCC in~$\mc{G}_W \otimes \mc{G}_A$ with similar partial order as in~$\mc{G}_A$. Then, from the structure of $D_C$ in Eq.~\eqref{D_H}, every giant parent SCC in $\mc{G}_W \otimes \mc{G}_A$ is associated with an output and, from Theorem~\ref{thm_scc}, the sufficient output-connectivity for observability holds. This completes the proof.
	\end{proof}
	Recall that Theorem~\ref{thm_main2} proves strong-connectivity of $\mc{G}_W$ as the \textit{sufficient} condition for distributed observability. This condition implies that the information of parent nodes/SCCs (e.g., nodes $\mb{x}_5$ and $\mb{x}_6$ in Fig.~\ref{fig_nca}) is measured and shared over the network as there exists a path from every sensor to every other sensor over $\mc{G}_W$. Otherwise, in case there is no path between two sensors (for example, from sensor $i$ to sensor $j$), the information of some state nodes (measured by sensor $i$) is not reachable to another sensor (for example, sensor $j$). However, this network connectivity condition can be more relaxed in particular cases by adding more sensor measurements. For example, if two sensors $i,j$ measure the same parent node/SCC, there is no need to have a path between them. In such cases, the strong-connectivity condition is only sufficient and not necessary.   
	
		\begin{rem}
			The above theorems prove $(W\otimes A, D_C)$-observability under certain conditions in structural sense. It is known that structural properties (including observability) are generic and hold for almost all numerical values of system parameters except for an algebraic subspace of zero Lebesgue measure \cite{woude:03}. This implies that assessing $(W\otimes A, D_C)$-observability via numerical test, i.e., Grammian rank, almost always gives the same result except for special conditions, e.g., ill-conditioned systems or when certain nodes/links fail. In case of link/node failure, the new structure of the system/consensus matrix needs to be considered for observability analysis. The structural results are widely used for observability analysis (instead of Grammian rank); for example, see \cite{pequito2015framework,ijss,pequito2013structured} for structural sensor placement based on system observability constraint.
	\end{rem}
	
	\subsection{Redundant Survivable Network Design} \label{secsub_redund}
	One key concept that underpins robust network architecture is \textit{redundant network design} or \textit{survivable network design} \cite{son2022situation}. This is based on the fact that one can add redundancy in terms of the number of links and connectivity of the ITS to preserve observability despite missing some information and losing some connections. These concepts rely heavily on the mathematical notions of $q$-node-connected graphs and $q$-link-connected graphs.
	
	\begin{defn}
		A graph is said to be $q$-node-connected (or $q$-vertex-connected) if it remains connected whenever fewer than $q$ nodes are removed. This property ensures that the network can sustain the failure of up to $q-1$ nodes without becoming disconnected. The parameter $q$ represents the level of redundancy in the network.
	\end{defn}
	\begin{defn}
		A graph is $q$-link-connected (or $q$-edge-connected) if it remains connected whenever fewer than $q$ links are removed, i.e., it can preserve (strong) connectivity in the face of up to $q-1$ link removal/failure. This concept focuses on the resilience of the network to the failure of connections (links) rather than nodes.       
	\end{defn}
	
	In this regard, redundant (or survivable) network design aims to create networks that have built-in redundancy to enhance resiliency and fault tolerance. Note that some works provide approximation algorithms to design redundant networks, see \cite{jabal2021approximation,diarrassouba2024optimization} for example. By designing networks that are $q$-node-connected or $q$-link-connected, redundancy is ensured, making the network more robust against failures. This (i) ensures that multiple independent paths exist between nodes to allow for message-passing (or data-sharing) in case of failures, and (ii) includes additional nodes and links as backups that can take over in case of primary component failures, e.g., to recover the loss of observability. In this direction,  \textit{Menger's theorem} is a relevant concept. This theorem is a fundamental result in graph theory that provides a precise characterization of the connectivity properties of graphs and it connects the concepts of node and link connectivity with the number of disjoint paths between nodes.
	
	\begin{thm}
		(Menger's theorem) For any two distinct nodes 
		$i$ and $j$ in the graph $\mc{G}_W$, 
		\begin{itemize} 
			\item let $\lambda(i,j)$ be the maximum number of link-disjoint paths (paths that do not share any links) between $i$ and $j$ and $q(i,j)$ be the minimum number of links whose removal disconnects $i$ from $j$. Then, $\lambda(i,j)=q(i,j)$.
			\item let $\lambda'(i,j)$ be the number of pairwise internally disjoint paths (paths that do not share any nodes) between $i$ and $j$ and $q'(i,j)$ be the minimum number of nodes whose removal disconnects $i$ from $j$. Then, $\lambda'(i,j)=q'(i,j)$.
		\end{itemize}  
	\end{thm}
	
	\begin{proof}
		The proof is given in \cite{modern_graph}.
	\end{proof}
	
	Menger's Theorem is significant since it provides a clear relationship between connectivity and path-disjointness in graphs. It helps in both investigating and designing $q$-link-connected ($q$-node-connected) networks by ensuring that there are enough disjoint paths to maintain connectivity despite link failures (or to handle node failures). Some example networks are investigated in Table~\ref{tab_q}.
	\begin{table} [h] 
		\centering
		\caption{Comparison between $q$-node connectivity and $q$-link connectivity of different graphs of size $n$. }
		\label{tab_q}
		\begin{tabular}{|c|c|c|} 
			\hline
			\hline
			Graph Type & Node connectivity & Link connectivity   \\
			\hline
			complete & $n$ &  $n$  
			\\   
			\hline
			cycle &  $2$ &  $2$   
			\\   
			\hline 
			star &  $1$    & $1$
			\\
			\hline
			path &  $1$    & $1$
			\\
			\hline
			$m$-nearest neighbour ring &  $2m$    & $2m$
			\\
			\hline
			\hline
		\end{tabular}
	\end{table}

	\begin{exm} \label{exm_mring}
		An $m$-nearest neighbour cycle (or ring) network is a type of graph topology where each node is connected to its 
		$m$ nearest neighbours in a circular (ring) arrangement. This topology generalizes the simple cycle (ring) network by allowing each node to have more than two direct connections, thereby increasing the redundancy and fault tolerance. An example $2$-nearest neighbour ring network of vehicles is shown in Fig.~\ref{fig_mring}.
		\begin{figure} 
			\centering
			\includegraphics[width=1.75in]{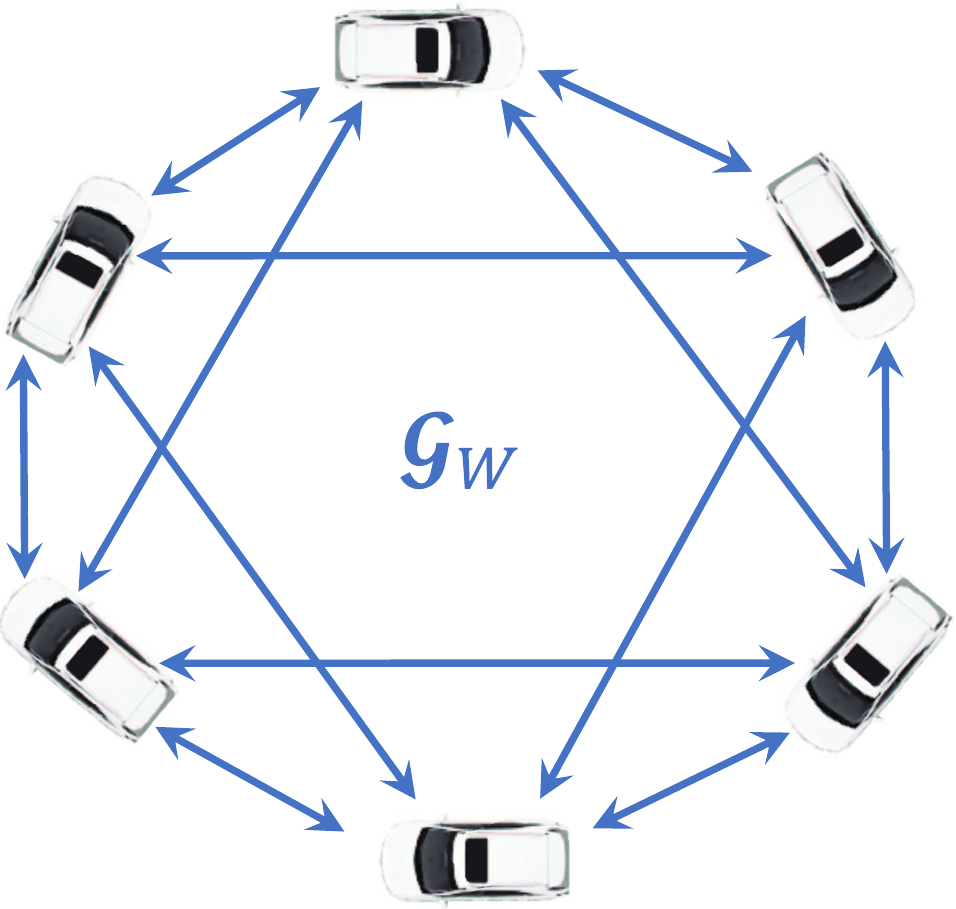}
			\caption{An example vehicle network of $2$-nearest neighour ring. 
			} \label{fig_mring}
		\end{figure}
		This network is $4$-node/link-connected, implying that it remains connected by removal of up to $3$ (randomly chosen) nodes/links.
	\end{exm}
	
	\subsection{Redundant Distributed Observer Design} \label{subsec_obsrv}
	In this section, a consensus-based single-time-scale redundant distributed observer is proposed. In contrast to the double-time-scale observer design in \cite{he2021secure,qian2022consensus,battilotti2021stability}, it only performs one step of consensus between two consecutive samples of ITS dynamics, i.e., it has no inner consensus loop. This significantly reduces the communication, data-sharing, and computation loads as compared to \cite{he2021secure,qian2022consensus,battilotti2021stability}. The observer includes one step of consensus on a priori estimates (or prediction) and one step of local observation-update (or innovation) as follows: 
	\begin{align}\label{eq_p} 
		\widehat{\mb{x}}_{k|k-1}^i &= \sum_{j\in\mathcal{N}_i} w_{ij}A\widehat{\mb{x}}^j_{k-1|k-1},
		\\ 
		\widehat{\mb{x}}_{k|k}^i &= \widehat{\mb{x}}_{k|k-1}^i +K_i C_i \left({y}_{i,k}-C_i^\top \widehat{\mb{x}}_{k|k-1}^i\right), \label{eq_m}
	\end{align}
	with $K_i$ as the local observer gain at node $i$ and matrix $W$ satisfying row-stochastic property for consensus. The observer gain matrix $K_i$ is locally designed for each CAV and, thus, the global matrix $K$ is block-diagonal. This feedback gain is designed by solving the following linear-matrix-inequality (LMI), which is based on the polynomial-order cone-complementary algorithms in \cite{rami:97}:
	\begin{equation} \label{eq_min}
		\begin{aligned}
			\displaystyle
			\min
			~~ &  \mathbf{trace}(XY) \\
			\text{s.t.}  ~~& X,Y\succ 0,\\ ~~ & \left( \begin{array}{cc} X&\widehat{A}^\top\\ \widehat{A}&Y\\ \end{array} \right) \succ 0,~ \left( \begin{array}{cc} X&I\\ I&Y\\ \end{array} \right) \succ 0,\\
			~~ & K\mbox{~is~block-diagonal}.\\
		\end{aligned}
	\end{equation} 
	where $\widehat{A} = W\otimes A - KD_C(W\otimes A)$.  
	The proposed distributed observer is summarized in Algorithm~\ref{alg_ac}.
	\begin{algorithm} 
		\KwData{ Vehicle dynamics matrix $A$,  $q$-node/link-connected network $\mc{G}_W$}
		\KwResult{Estimated state $\widehat{\mb{x}}^i_{k|k}$ }
		{\textbf{Initialization:} $k=1$\;
			Design the row-stochastic consensus matrix $W$ associated with $\mc{G}_W$ (simply set $w_{ij}=\frac{1}{|\mc{N}_i|}$ for all $j \in \mc{N}_i$)\; 
			Design $K$ via LMI \eqref{eq_min}\;
		}
		
		\While{monitoring the ITS states\;
		}{Vehicle $i$ receives $\widehat{\mb{x}}^j_{k-1|k-1}$ from neighboring vehicles $j\in \mc{N}_i$\;
			Vehicle $i$ updates its state $\widehat{\mb{x}}_{k|k}^i$ via dynamics \eqref{eq_p}-\eqref{eq_m}\;
			$k\leftarrow k+1$\;
		}
		\caption{\textsf{The Distributed Redundant Observer}} \label{alg_ac}
	\end{algorithm} 
	
	Define $\mb{e}_{k}^i:=  \mb{x}_{k}-\widehat{\mb{x}}_{k|k}^i$ as the observer error at the $i$th CAV. Then, mixing the observer dynamics~\eqref{eq_p}-\eqref{eq_m} with system dynamics~\eqref{eq_sys1} and output model~\eqref{eq_H_i}, it follows that,
	\begin{align}
		\mb{e}_{k}^i =& \sum_{j\in \mathcal{N}_i}w_{ij}A(\mb{x}_{k-1} - \widehat{\mb{x}}^j_{k-1|k-1})  \nonumber \\ \label{eq_ei}
		&- K_i C_i^\top C_i \sum_{j\in \mathcal{N}_i}w_{ij}A(\mb{x}_{k-1} - \widehat{\mb{x}}^j_{k-1|k-1})+ \zeta_k
	\end{align}
	with $\zeta_k$ collecting the noise-related terms. In compact form, define the global error at all nodes as, 
	\begin{align}\nonumber
		\mb{e}_{k} &= \left( \begin{array}{c}
			\mb{e}_{k}^1\\
			\vdots \\
			\mb{e}_{k}^n
		\end{array}\right).
	\end{align}
	Then, one can find the global error dynamics by concatenating the local dynamics given by \eqref{eq_ei} as follows:
	\begin{align}\label{eq_err1}
		\mb{e}_{k} = (W\otimes A - KD_C(W\otimes A))\mb{e}_{k-1} +
		\zeta_k = \widehat{A}\mb{e}_{k-1} +
		\zeta_k ,
	\end{align}
	For Schur stability of Eq.~\eqref{eq_err1}, following from Kalman stability theorem \cite{bay}, we need $(W\otimes A,D_C)$-observability. This verifies the formulation of distributed observability by \eqref{eq_dist_obsrv} in Section~\ref{sec_prob}. Therefore, from the results of Section~\ref{secsub_main}, the distributed observability is satisfied by:
	\begin{itemize}
		\item taking one output from every parent SCC in $\mc{G}_A$ (i.e., centralized observability\footnote{Recall that centralized observability is a necessary condition for distributed observability. In other words, if the pair $(A,C)$ is not  observable, there is no network $\mc{G}_W$ to satisfy $(W\otimes A,D_C)$-observability.}), and
		\item designing network $\mc{G}_W$ as an SC graph topology, i.e., an irreducible structure for $W$ matrix (while satisfying row-stochasticity).
	\end{itemize}
	
	In fact, the LMI design~\eqref{eq_min} provides a block-diagonal gain $K=\mbox{diag}[K_i]$ to ensure the Schur stability of the error dynamics~\eqref{eq_err1} (i.e., Schur stability of matrix $\widehat{A}$). Although in a centralized design with general (not necessarily block-diagonal) $K$, the LMI can be augmented with the contraction factor for an exponential decay rate of the error dynamics, in the distributed case with block-diagonal constraint on $K$ it is not straightforward to restrict the spectral radius of $\widehat{A}$ via convex LMI constraints to impose a decay rate of the error dynamics. Such a constrained LMI design for $K=\mbox{diag}[K_i]$ with pole-region constraints (to address settling-time, overshoot, etc.) is still an open problem and one direction of our future research. 
	
	\begin{rem}
		The proposed observer can be made resilient by the redundant design of $\mc{G}_W$ network as discussed in Section~\ref{secsub_redund}. In this direction, for $q$-redundant distributed observability, two conditions need to be satisfied:
		\begin{itemize}
			\item taking output of $q$ number of states from every parent SCC in $\mc{G}_A$, and
			\item designing a $q$-node/link-connected network topology  $\mc{G}_W$.
		\end{itemize}
	\end{rem}

		Table~\ref{tab_compare} compares the proposed distributed estimator with other existing distributed estimation techniques. As compared with \cite{chen2023distributed,WANG20174039,kar2013consensus,das2015distributed,chen2018internet}, no local observability is assumed in our work, which reduces the connectivity requirement on the communication network. As compared with double time-scale scenarios \cite{he2021secure,qian2022consensus,battilotti2021stability,ghods2025resilient}, the proposed estimator \eqref{eq_p}-\eqref{eq_m} only performs one iteration of communication and consensus computation per sample of system dynamics. This is in contrast to $L$ iterations of communication and consensus between two samples $k$ and $k+1$ of system dynamics (with $L$ more than network diameter) in \cite{he2021secure,qian2022consensus,battilotti2021stability,ghods2025resilient}. Therefore, our proposed estimator reduces the communication/computation load on CAVs for practical applications.
	
	\begin{table*} [bpt!] 
		\centering
			\caption{Comparing distributed estimation methods in terms of observability assumption, redundant design, and scale of communication and consensus (computation) iterations for a CAV network of size $n$.  }
			\label{tab_compare}
			\begin{tabular}{|c|c|c|c|c|} 
				\hline
				Literature &  observability & communication & consensus & redundancy  \\
				\hline
				this work  & global-$(A,C)$ &  $n \times 1$ & $1$ & $\checkmark$
				\\   
				\hline
				\cite{chen2023distributed,WANG20174039,kar2013consensus,das2015distributed,chen2018internet} &  local-$(A,C_i)$ &  $3n \times 1$   & $1$ & -
				\\   
				\hline 
				\cite{he2021secure,qian2022consensus,battilotti2021stability,ghods2025resilient} &  global-$(A,C)$ &  $n \times L$    & $L$ & - \\
				\hline 
				\hline
		\end{tabular}
	\end{table*}
	
		\begin{rem}
			It is known that the cone-complementary algorithms for LMI gain design (as in \eqref{eq_min}) are of polynomial-order complexity \cite{nesterov1994interior,ye1993fully}. Moreover, the dynamics \eqref{eq_p}-\eqref{eq_m} is of polynomial-order complexity of $\mc{O}(n^2N^2m^2)$. These imply that the distributed solution can be scaled up computationally for large-scale setups.
			The communication complexity of the proposed estimator is $\mc{O}(n)$ with $n$ as the number of CAVs. This is due to the fact that a strongly-connected network $\mc{G}_W$ (with minimum number of $n$ links) is sufficient for the distributed estimation design, implying that the solution can be scaled up in terms of communication requirement for large-scale networks.
	\end{rem}
	
	\section{Illustrative Example and Simulations} \label{sec_examp}

	In this section, we consider the mixed traffic scenario of $4$ HDVs and $5$ CAVs given in Fig.~\ref{fig_platoon} in which some  CAVs have sensor measurements from the preceding HDVs and some are without sensor measurements, obtaining information via shared data from other CAVs. 
	
	For the simulation of HDVs dynamics, two different models are considered depending on their position in the mixed traffic: a linear free-flow model \cite{ahmed1999modeling} or Helly’s linear car-following model \cite{brackstone1999car}. If the HDV is an ego
		vehicle or its distance from the front vehicle is large enough (e.g., HDVs $3,4$ in Fig.~\ref{fig_platoon}), its dynamics is described by the free-flow model. In other words, for HDVs at a safe distance from their front vehicle (i.e., distance more than a threshold distance $L$), the free-flow equation of motion is considered. However, if the distance to the front vehicle is less than a distance threshold $L$, then the car-following model is adopted to describe its equation of motion (e.g., HDVs $1,2$ in Fig.~\ref{fig_platoon}). The discrete-time version of linear free-flow dynamics is given as follows:
	\begin{align}  \label{eq_free-flow}
		v(k+1) = v(k) + \lambda (v_d(k) - v(k-\tau)) + \epsilon(k)
	\end{align}
	where $v(k)$ is the HDV velocity at the time-step $k$, $\tau$ is a factor denoting the reaction time
	of the HDV, $\lambda$ is a coefficient describing
	how fast the vehicle can track the desired velocity $v_d(k)$,
	and $\epsilon(k) \sim \mc{N}(0,\sigma)$ denoting the zero-mean Gaussian noise.
	
	Helly’s linear car-following dynamics \cite{brackstone1999car} is modeled in discrete-time as:
	\begin{align}  \label{eq_helly}
		v(k+1) = v(k) + \alpha_1 \delta v(k-\tau)+ \alpha_2 (\delta x(k-\tau) -D(k))
	\end{align} 
	where $\delta v(k)$ and $\delta x(k)$ respectively denote the difference of the speed and positions of the HDV and its front vehicle, $D(k) = \beta_1 + \beta_2 v(k-\tau)$ denotes the desired distance, $\alpha_1$ and $\alpha_2$ are constants, $\beta_1$ and $\beta_2$ denote the coefficients
	related to desired distance headway, and $\tau$ is the reaction time constant.  The simulation parameters are given in Table~\ref{tab_sim}.
	\begin{table} [h] 
		\centering
		\caption{The simulation parameters. }
		\label{tab_sim}
		\begin{tabular}{|c|c|c|c|c|c|} 
			\hline
			$\lambda$ & $0.3$  
			&
			$\tau$ &  $10$ &     
			$\alpha_1$ &  $0.5$ \\   
			\hline
			$\alpha_2$ &  $0.125$&
			$\beta_1$ &  $4$  &  
			$\beta_2$ &  $0.05$
			\\
			\hline
			\hline
		\end{tabular}
	\end{table}   
	
	In this simulation, the purpose is to enable all CAVs to estimate the positions and velocities of all HDVs by sharing relevant information over the communication network $\mc{G}_W$. Each CAV considers a nearly-constant-velocity (NCV) model for the $i$-th HDV  as $\mb{x}_{i,k+1}=\tilde{A}_i\mb{x}_{i,k}+\nu_{i,k}$ with $\mb{x}_{i}=[p_{i,x},v_{i,x}]^\top$ and
	\begin{equation}
		\tilde{A}_i=\left(
		\begin{array}{cc}
			1 & T \\
			0 & 1  
		\end{array} \right).
	\end{equation}
	The network $\mc{G}_W$ of CAVs given in Fig.~\ref{fig_platoon} is $1$-node/link-connected, i.e., it is resilient to the isolation/removal of $1$ communication link or sensor node. The entries of its associated adjacency matrix $W$ are set as $\frac{1}{|\mc{N}_i|} = \frac{1}{3}$ which satisfies row-stochasticity, i.e., $\sum_{j=1}^n w_{ij}=1$ (recall that $\mc{N}_i$ includes two neighbouring CAVs and CAV $i$ itself). The shared information includes the state (velocity and position) estimates of HDVs, denoted by $\widehat{\mb{x}}^j_{k-1|k-1}$ in Algorithm~\ref{alg_ac}, among neighbouring CAVs. We adopt the redundant observer design in Algorithm~\ref{alg_ac} to enable every CAV to locally estimate the state of all HDVs with bounded steady-state error. The spectral redius of $W \otimes A$ before the gain design is $1$ and after LMI gain design \eqref{eq_min} the spectral radius of $\widehat{A}$ is $0.974$ which is Schur stable. The \textit{stable} steady-state error is partially due to measurement noise in the model and uncertainty in the model dynamics (e.g., due to $\epsilon(k)$). The measurement noise is considered as $\mu_{i,k} \sim \mc{N}(0,0.1)$.  In the scenario considered, HDV 3, HDV 1 and HDV 2, change their desired velocities at $t=25\rm{s}$, in which HDV 3 and HDV 2  increase their velocity to $30 \rm{m/s}$ and HDV 1 reduces its velocity to $20 \rm{m/s}$. Fig.~\ref{fig_mse1}  and Fig.~\ref{fig_mse2}  present the positions and velocities of HDVs and the estimated velocities by all $5$ CAVs. From the figure, it can be seen that all CAVs reach consensus on the velocity of HDVs and track the HDV velocity with steady-state bounded error (which is due to the noise in the model). 
	
	\begin{figure} 
		\centering
		\includegraphics[width=4.7in]{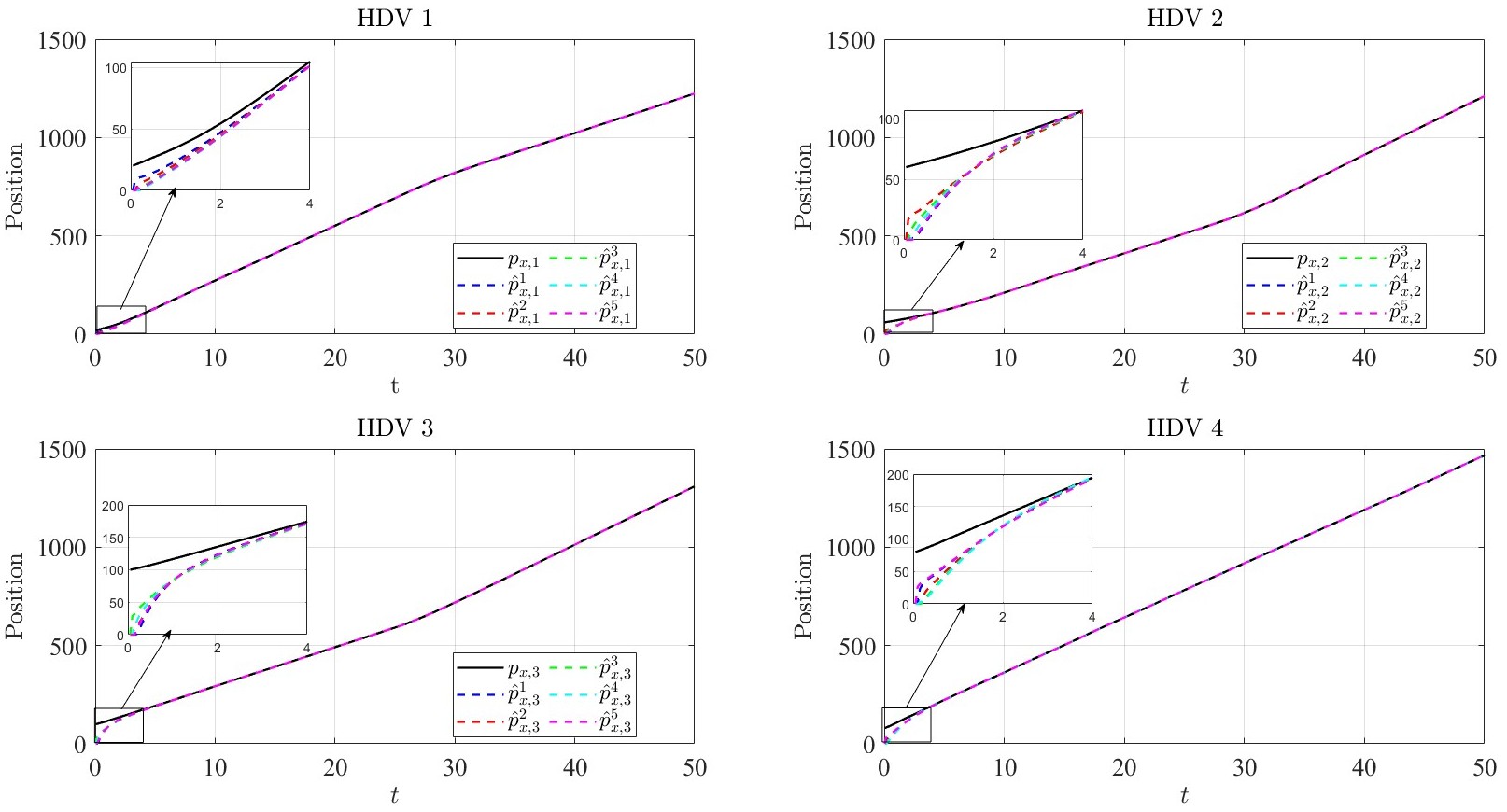}
		\caption{The position of the $i$-th HDV, $p_{x,i}$, $i=1,\dots,4$ and the estimated one, $\hat{p}^j_{x,i}$ by the $j$-th CAV $j=1,\dots,5$ using the distributed observer in Algorithm~\ref{alg_ac}. 
		} \label{fig_mse1}
	\end{figure}
	\begin{figure} 
		\centering
		\includegraphics[width=4.7in]{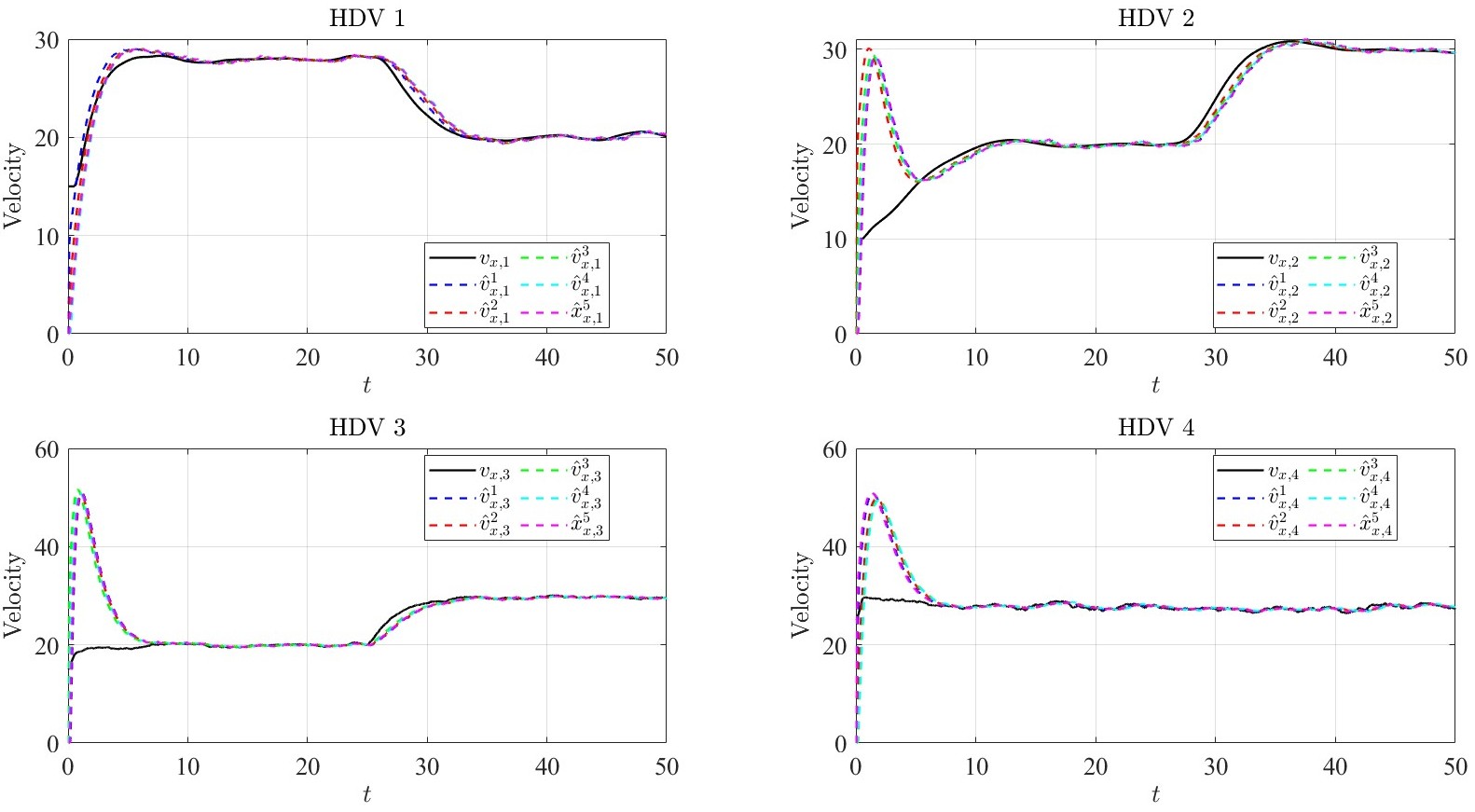}
		\caption{The velocity of the $i$-th HDV, $v_{x,i}$, $i=1,\dots,4$ and the estimated one, $\hat{v}^j_{x,i}$ by the $j$-th CAV $j=1,\dots,5$ using the distributed observer in Algorithm~\ref{alg_ac}.
		} \label{fig_mse2}
	\end{figure}
	
	For the next simulation, we check the resiliency of our redundant observer to link failure. Assume that the link between CAV $1$ and $2$ in Fig.~\ref{fig_platoon} is corrupted due to environmental conditions. Since the network $\mc{G}_W$ is $1$-link-connected, the remaining network after link removal is still strongly-connected. Therefore, from Theorem~\ref{thm_main2}, distributed observability still holds and the error dynamics \eqref{eq_err1} is Schur stabilizable. Note that in the new network setup the CAVs $1$ and $2$ renew their stochastic link weights to $\frac{1}{|\mc{N}_i|} = \frac{1}{2}$. We apply the proposed Algorithm~\ref{alg_ac} to the new setup, and the simulation results are shown in Fig.~\ref{fig_mse_remov1} and Fig.~\ref{fig_mse_remov2}. The spectral redius of $\widehat{A}$ via LMI gain design in \eqref{eq_min} is $0.973$ for this case.  Clearly from the figure, the CAVs are able to track the state  of the HDVs with bounded error and also reach consensus on the state variables. 
	
	\begin{figure} 
		\centering
		\includegraphics[width=4.7in]{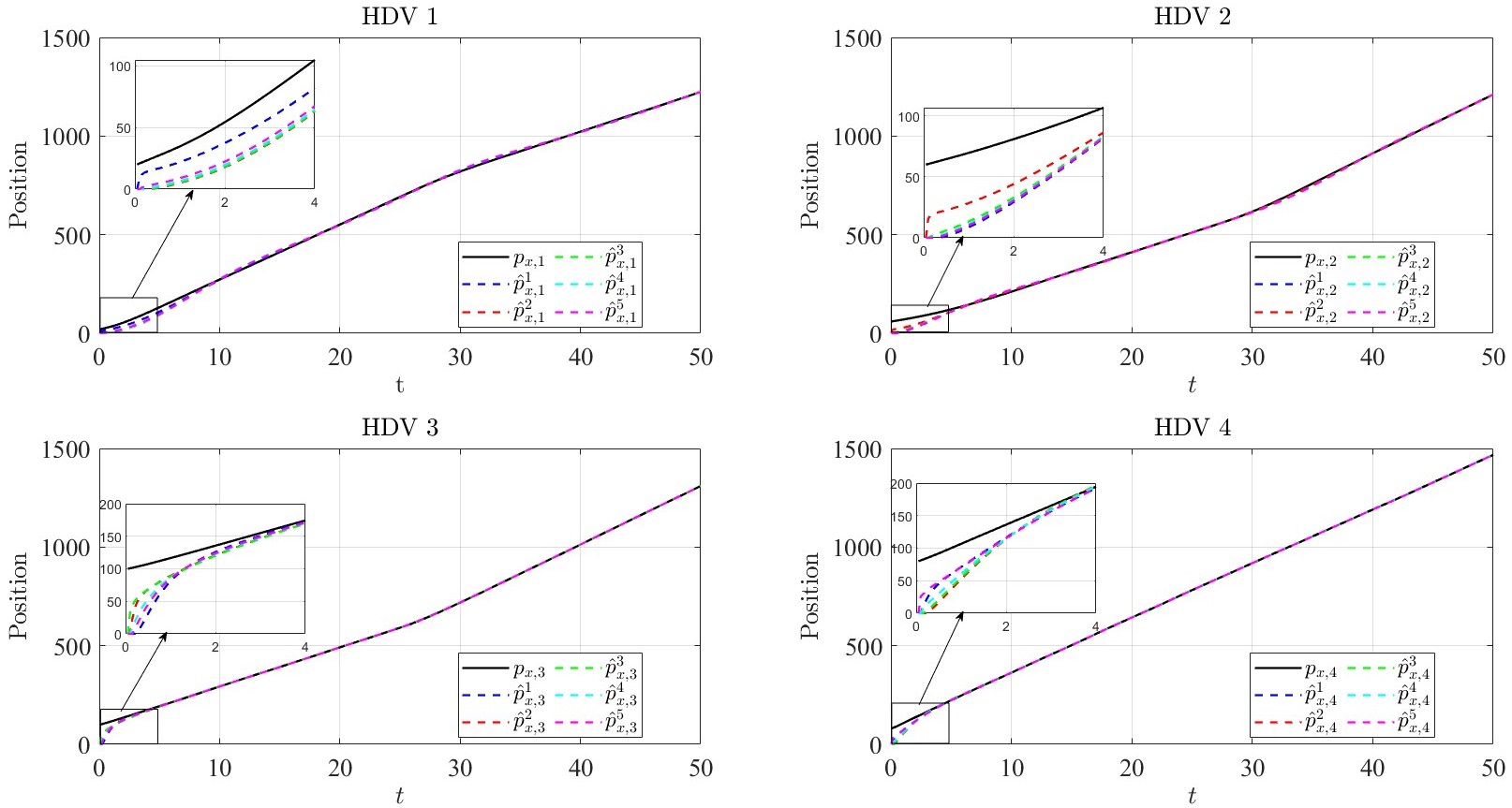}
		\caption{The position of the $i$-th HDV, $p_{x,i}$, $i=1,\dots,4$ and the estimated one, $\hat{p}^j_{x,i}$ by the $j$-th CAV $j=1,\dots,5$ using the distributed observer in Algorithm~\ref{alg_ac} after removing the corrupted link from the communication network of CAVs. 
		} \label{fig_mse_remov1}
	\end{figure}
	\begin{figure} 
		\centering
		\includegraphics[width=4.7in]{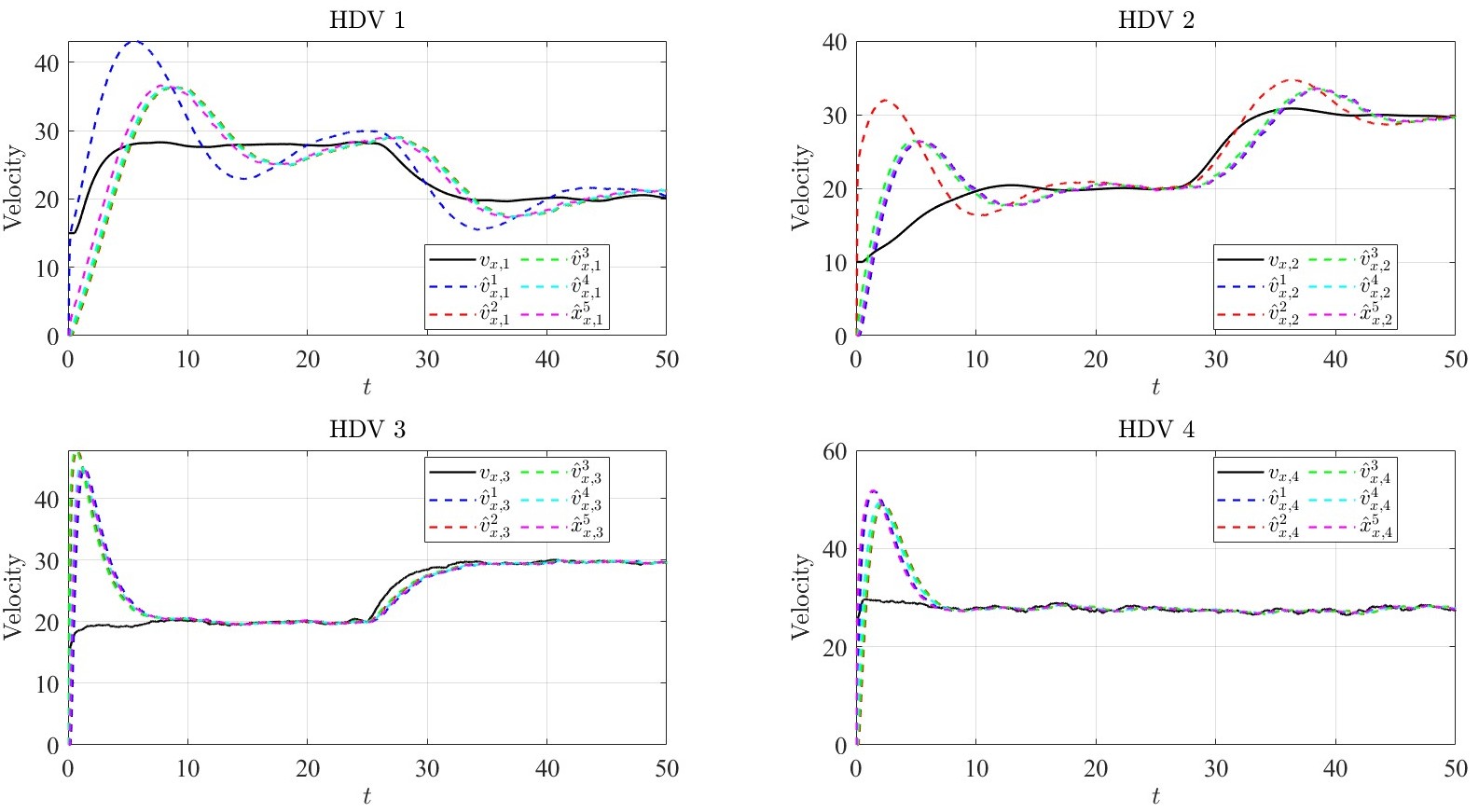}
		\caption{The velocity of the $i$-th HDV, $v_{x,i}$, $i=1,\dots,4$ and the estimated one, $\hat{v}^j_{x,i}$ by the $j$-th CAV $j=1,\dots,5$ using the distributed observer in Algorithm~\ref{alg_ac} after removing the corrupted link from the communication network of CAVs. 
		} \label{fig_mse_remov2}
	\end{figure}
	
	\begin{figure} 
		\centering
		\includegraphics[width=3.5in]{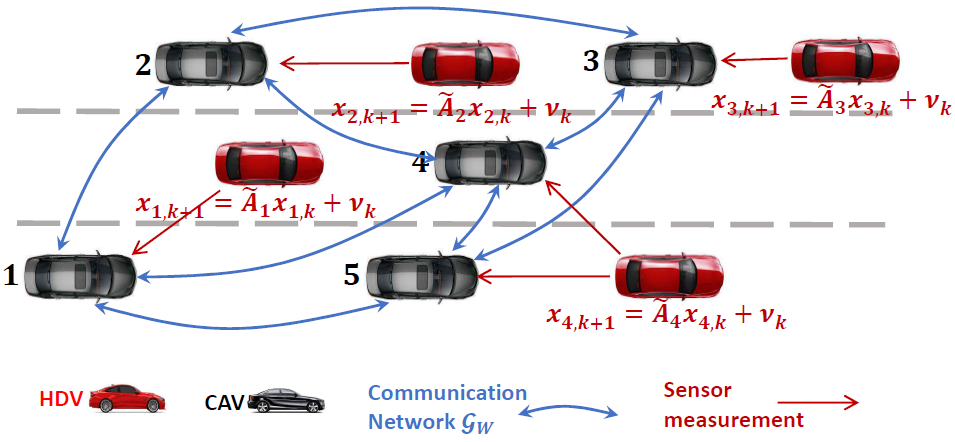}
		\caption{This figure considers the mixed traffic network of CAVs and HDVs with more communications among CAVs to add more redundancy. The network of CAVs is $2$-node/link-connected, implying resiliency to removal of up-to $2$ nodes/links. 
		} \label{fig_platoon_redund}
	\end{figure} 
	
	For the next simulation we consider the redundant network $\mc{G}_W$ of CAVs given in Fig.~\ref{fig_platoon_redund}, which is $2$-node/link-connected. 
		The simulation parameters of the free-flow model at HDVs $3,4$ and the car-following model at HDVs $1,2$ are given in Table~\ref{tab_sim2}.
		\begin{table} [h] 
			\centering
				\caption{The simulation parameters. }
				\label{tab_sim2}
				\begin{tabular}{|c|c|c|c|c|c|} 
					\hline
					$\lambda$ & $0.4$  
					&
					$\tau$ &  $15$ &     
					$\alpha_1$ &  $0.4$ \\   
					\hline
					$\alpha_2$ &  $0.15$&
					$\beta_1$ &  $10$  &  
					$\beta_2$ &  $0.5$
					\\
					\hline
					\hline
			\end{tabular}
		\end{table}
		Similar to the previous case, CAVs approximate the dynamics of HDVs with the NCV model. The measurement noise is  $\mu_{i,k} \sim \mc{N}(0,0.15)$. More complex velocity variation for HDVs is considered for this scenario as shown in Fig.~\ref{fig_redund1} and~\ref{fig_redund2}. The CAVs share velocity/position estimates  $\widehat{\mb{x}}^j_{k-1|k-1}$ according to Algorithm~\ref{alg_ac} and update their estimates using the data received from their neighbours based on dynamics~\eqref{eq_p}-\eqref{eq_m}. The spectral radius of $\widehat{A}$ via LMI gain design in \eqref{eq_min} is $0.974$ for this case. The estimated position and velocity of HDVs at all CAVs are shown in Fig.~\ref{fig_redund1} and~\ref{fig_redund2}, respectively. We compare the mean square estimation error (MSEE) of the position and velocity estimates with the centralized Kalman filter as the benchmark model, see Fig.~\ref{fig_redund_comp} and Fig.~\ref{fig_redund_comp2}, respectively. Note that for the centralized Kalman filter, all sensor measurements are provided to a central data fusion unit, while in our distributed scenario, each CAV has access to only one local sensor measurement and estimates the position/velocity of HDVs based on the data shared over the network over time. In other words, some measurements are not available instantly, and the CAVs get this information over time from the neighbouring CAVs.  This is the reason behind the larger MSEE in the distributed case.   
		\begin{figure} 
			\centering
			\includegraphics[width=4.7in]{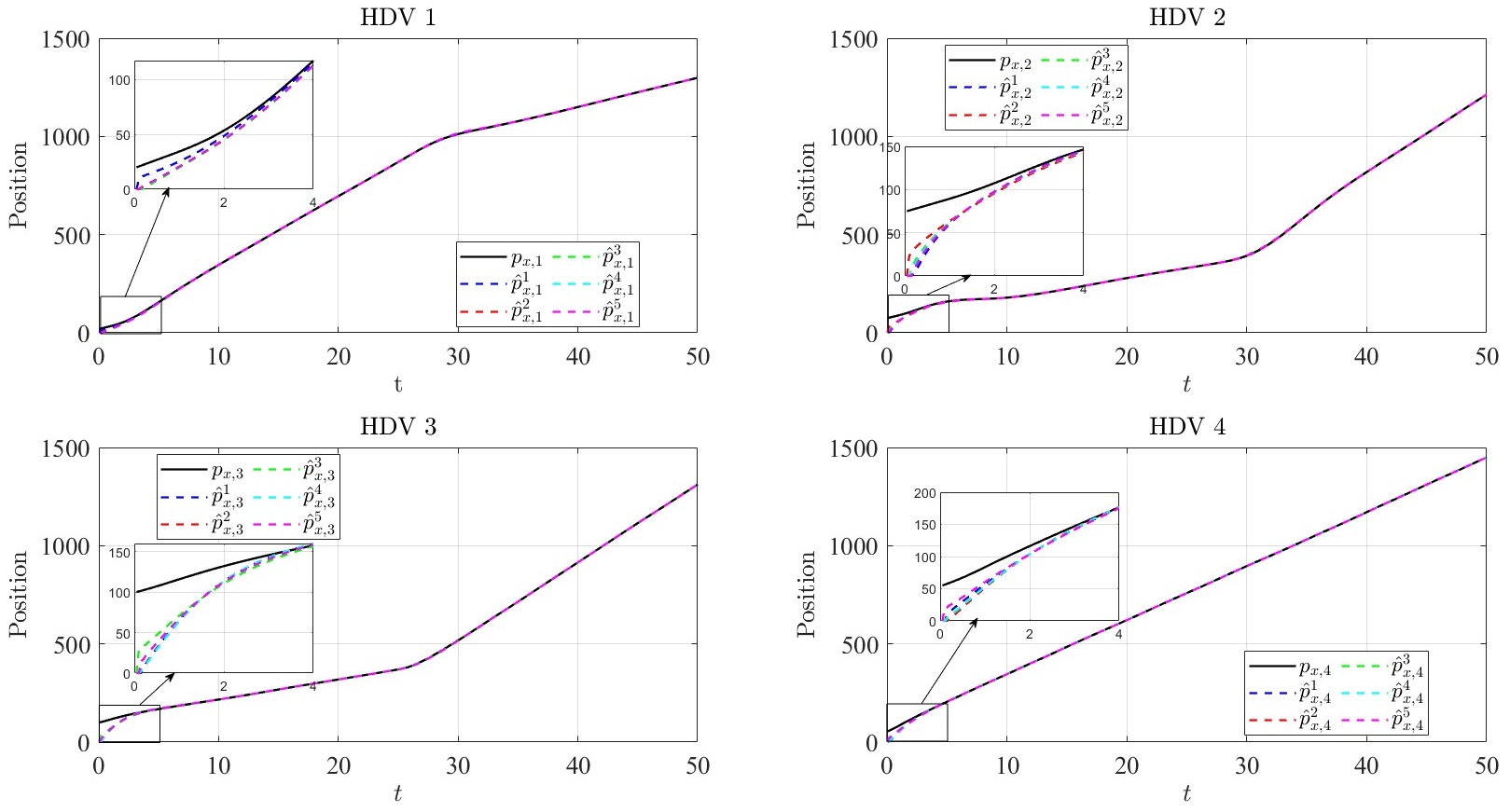}
			\caption{The position of the $i$-th HDV, $p_{x,i}$, $i=1,\dots,4$ and the estimated position, $\hat{p}^j_{x,i}$ by the $j$-th CAV $j=1,\dots,5$ using the distributed observer in Algorithm~\ref{alg_ac} over the redundant network in Fig.~\ref{fig_platoon_redund}.
			} \label{fig_redund1}
		\end{figure}
		\begin{figure} 
			\centering
			\includegraphics[width=2.7in]{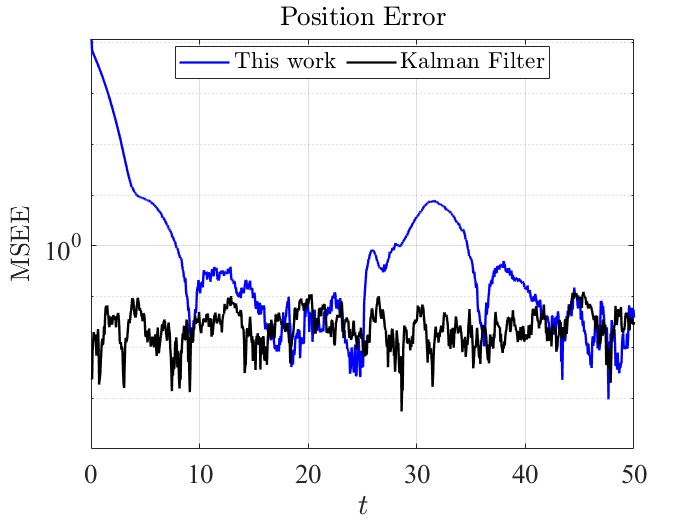}
			\caption{This figure compares the mean square estimation error (MSEE) of position estimates at all CAVs for the proposed distributed estimator and benchmark centralized Kalman filter. The HDV positions over time are presented in Fig.~\ref{fig_redund1}. 
			} \label{fig_redund_comp}
		\end{figure}
		\begin{figure} 
			\centering
			\includegraphics[width=4.7in]{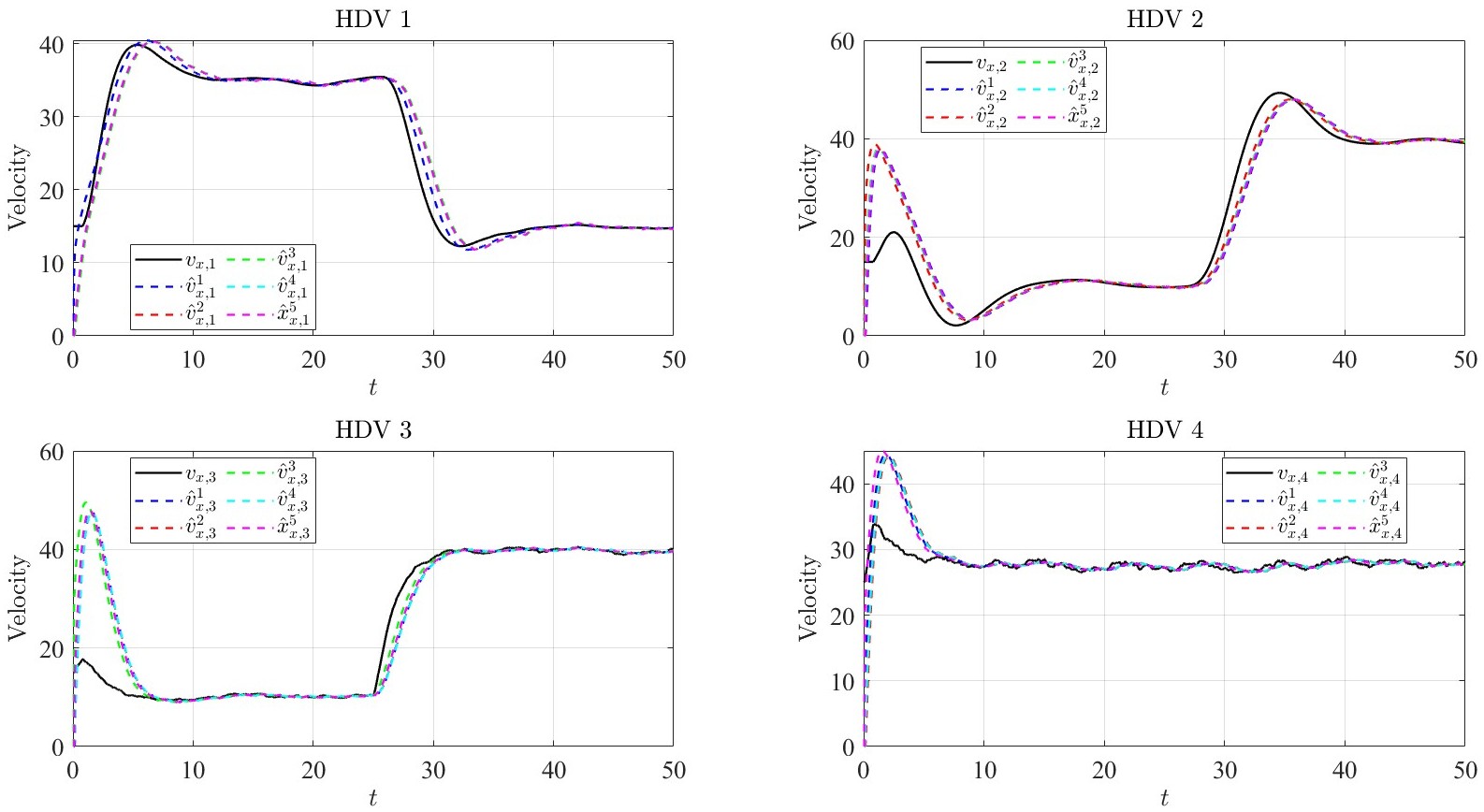}
			\caption{The velocity of the $i$-th HDV, $v_{x,i}$, $i=1,\dots,4$ and the estimated velocity, $\hat{v}^j_{x,i}$ by the $j$-th CAV $j=1,\dots,5$ using the distributed observer in Algorithm~\ref{alg_ac} over the redundant network in Fig.~\ref{fig_platoon_redund}.
			} \label{fig_redund2}
		\end{figure}
		\begin{figure} 
			\centering
			\includegraphics[width=2.7in]{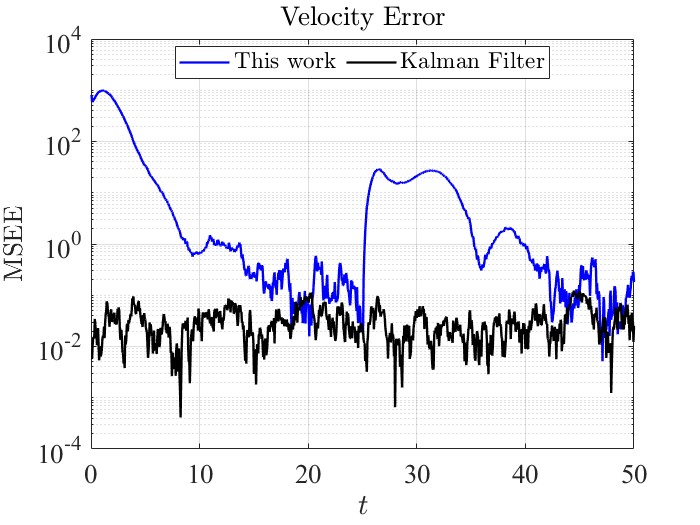}
			\caption{This figure compares the mean square estimation error (MSEE) of velocity estimates at all CAVs for the proposed distributed estimator and benchmark centralized Kalman filter. The HDV velocities over time are presented in Fig.~\ref{fig_redund2}.
			} \label{fig_redund_comp2}
		\end{figure}

		Next, to illustrate the resiliency to node/link failure, we remove CAV node $5$ and the links $(2,4)$ and $(2,3)$. Recall that as explained in Fig.~\ref{fig_platoon_redund}, the network $\mc{G}_W$ is $2$-node/link-connected. We repeat the simulations with the same parameters for this case. The gain matrices at $4$ CAVs are calculated via LMI~\eqref{eq_min} as follows:
		\begin{equation} \nonumber
			K_1=\left(
			\begin{array}{cccccccc}
				0.224 &    0.224&  0 &    0 & 0 & 0  &  0&    0\\
				0.223  &  0.223 &   0&   0&  0  & 0 & 0&    0 \\
				0  & 0  &  0.225  &  0.225  & 0 & 0 & 0&    0\\
				0  & 0  &  0.224  &  0.224  & 0 & 0 & 0&    0\\
				0  &  0  &  0  &  0  &  0  &  0 &   0 &   0 \\
				0  &  0  &  0  &  0  &  0  &  0 &   0 &   0 \\
				0  &  0  &  0  &  0  &  0  &  0 &   0.228 &   0.228 \\ 
				0  &  0  &  0  &  0  &  0  &  0 &   0.227 &   0.227 \\
			\end{array} \right).
		\end{equation}
		\begin{equation} \nonumber
			K_2=\left(
			\begin{array}{cccccccc}
				0.226 &    0.226&  0 &    0 & 0 & 0  &  0&    0\\
				0.225  &  0.225 &   0&   0&  0  & 0 & 0&    0 \\
				0  & 0  &  0.226  &  0.226  & 0 & 0 & 0&    0\\
				0  & 0  &  0.225  &  0.225  & 0 & 0 & 0&    0\\
				0  &  0  &  0  &  0  &  0.230  &  0.230 &   0 &   0 \\
				0  &  0  &  0  &  0  &  0.229  &  0.229 &   0 &   0 \\
				0  &  0  &  0  &  0  &  0  &  0 &   0 &   0 \\ 
				0  &  0  &  0  &  0  &  0  &  0 &   0 &   0 \\
			\end{array} \right).
		\end{equation}
		\begin{equation} \nonumber
			K_3=\left(
			\begin{array}{cccccccc}
				0 &    0&  0 &    0 & 0 & 0  &  0&    0\\
				0  &  0 &   0&   0&  0  & 0 & 0&    0 \\
				0  & 0  &  0.226  &  0.226  & 0 & 0 & 0&    0\\
				0  & 0  &  0.225  &  0.225  & 0 & 0 & 0&    0\\
				0  &  0  &  0  &  0  &  0.222  &  0.222 &   0 &   0 \\
				0  &  0  &  0  &  0  &  0.221  &  0.221 &   0 &   0 \\
				0  &  0  &  0  &  0  &  0  &  0 &   0.221 &   0.221 \\ 
				0  &  0  &  0  &  0  &  0  &  0 &   0.220 &   0.220 \\
			\end{array} \right).
		\end{equation}
		\begin{equation} \nonumber
			K_4=\left(
			\begin{array}{cccccccc}
				0.230 &    0.230 &  0 &    0 & 0 & 0  &  0&    0\\
				0.230  &  0.230 &   0&   0&  0  & 0 & 0&    0 \\
				0  & 0  &  0  &  0  & 0 & 0 & 0&    0\\
				0  & 0  &  0  &  0  & 0 & 0 & 0&    0\\
				0  &  0  &  0  &  0  &  0.228  &  0.228 &   0 &   0 \\
				0  &  0  &  0  &  0  &  0.227  &  0.227 &   0 &   0 \\
				0  &  0  &  0  &  0  &  0  &  0 &   0.228 &   0.228 \\ 
				0  &  0  &  0  &  0  &  0  &  0 &   0.227 &   0.227 \\
			\end{array} \right).
		\end{equation}
		The spectral radius of $\widehat{A}$ is $0.975$ for this case.
		The estimated velocities and positions at the remaining $4$ CAVs are represented in Figs.~\ref{fig_redund_remov1} and~\ref{fig_redund_remov2}, respectively. The comparison with the benchmark centralized Kalman filter is also provided in Figs.~\ref{fig_redund_remov_comp} and~\ref{fig_redund_remov_comp2}. 
		\begin{figure} 
			\centering
			\includegraphics[width=4.7in]{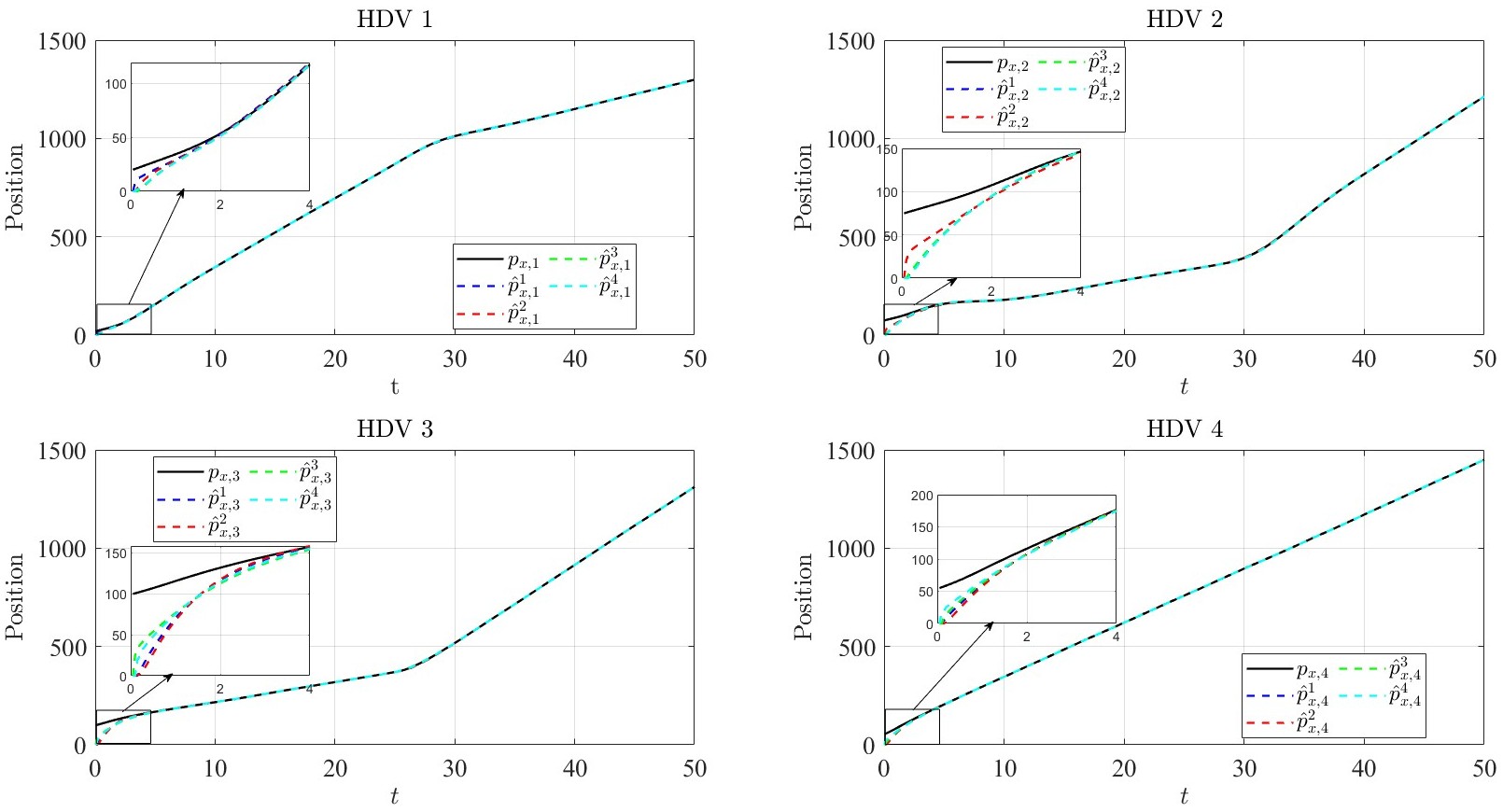}
			\caption{The position of the $i$-th HDV, $p_{x,i}$, $i=1,\dots,4$ and the estimated one, $\hat{p}^j_{x,i}$ by the $j$-th CAV $j=1,\dots,5$ using the distributed observer in Algorithm~\ref{alg_ac} after removing one corrrupted node and two corrupted link from the network of CAVs. 
			} \label{fig_redund_remov1}
		\end{figure}
		\begin{figure} 
			\centering
			\includegraphics[width=2.7in]{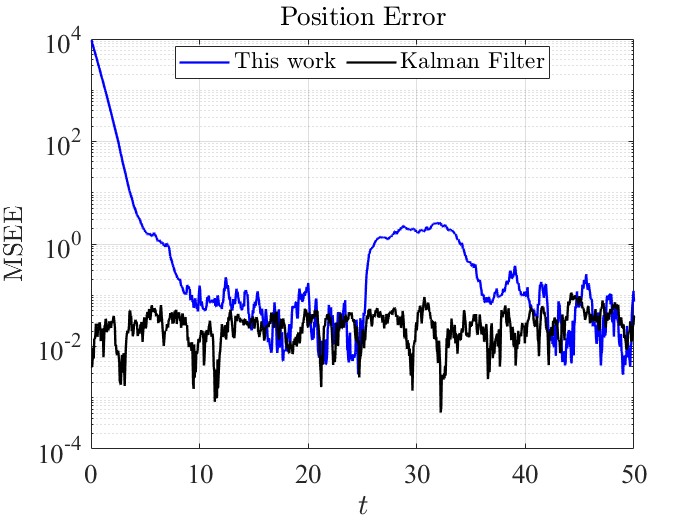}
			\caption{This figure compares the mean square estimation error (MSEE) of position estimates at all CAVs for the proposed distributed estimator and benchmark centralized Kalman filter. The HDV positions over time are presented in Fig.~\ref{fig_redund_remov1} for the resilient case.  
			} \label{fig_redund_remov_comp}
		\end{figure}
		\begin{figure} 
			\centering
			\includegraphics[width=4.7in]{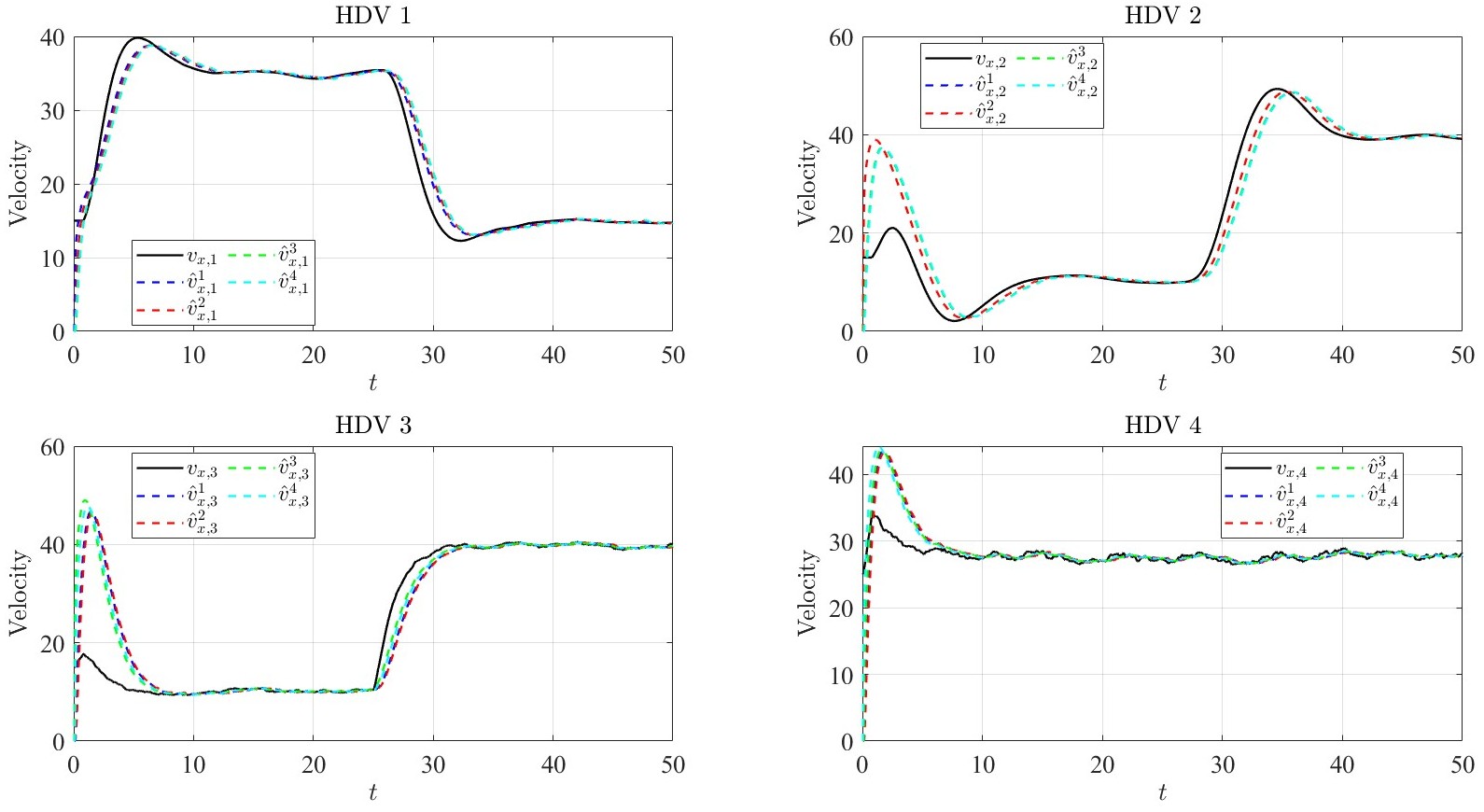}
			\caption{The velocity of the $i$-th HDV, $v_{x,i}$, $i=1,\dots,4$ and the estimated one, $\hat{v}^j_{x,i}$ by the $j$-th CAV $j=1,\dots,5$ using the distributed observer in Algorithm~\ref{alg_ac} after removing one corrrupted node and two corrupted link from the network of CAVs. 
			} \label{fig_redund_remov2}
		\end{figure}
		\begin{figure} 
			\centering
			\includegraphics[width=2.7in]{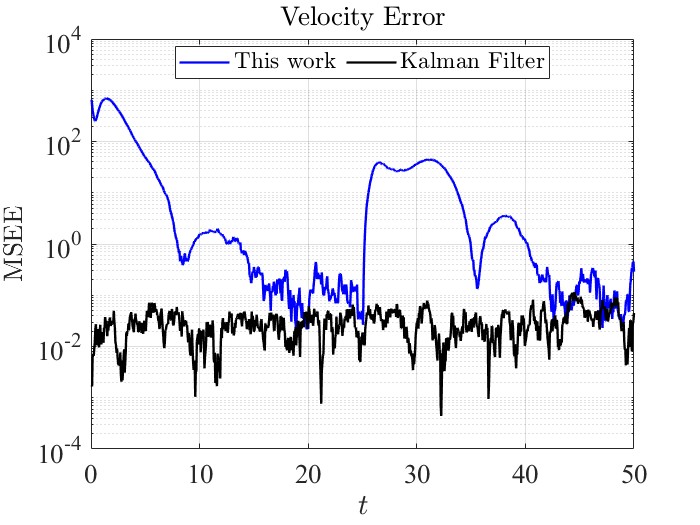}
			\caption{This figure compares the mean square estimation error (MSEE) of velocity estimates at all CAVs for the proposed distributed estimator and benchmark centralized Kalman filter. The HDV velocities over time are presented in Fig.~\ref{fig_redund_remov2} for the resilient case. 
			} \label{fig_redund_remov_comp2}
		\end{figure}

	\section{Conclusion} \label{sec_con}
	This paper derives a distributed observable state-space model for mixed traffic analysis and network of connected vehicles. The proposed formulation integrates vehicles' dynamic states, sensing capabilities, and their communication networks into one compatible framework, facilitating observability analysis for mixed traffic ITS. By establishing the conditions for distributed observability, strong network connectivity coupled with a block-diagonal observer gain ensures that each HDV's dynamics is observable to every other CAV through network-based data-sharing. Some of these CAVs may have no direct sensor measurement from the HDVs. We designed a distributed observer via locally sharing estimates and observations among neighboring CAVs.
	Moreover, we address the challenge of faulty sensors and unreliable observation data by introducing the concept of redundant distributed observability. The $q$-node/link-connected network design enhances system resilience, allowing for the isolation of a certain number of faulty sensors or corrupted communication links without compromising the distributed observability of the mixed traffic ITS. In this regard, our results advance the topology design of vehicle platoons and other types of intelligent transportation networks in terms of resiliency to faults, failures, or attacks. In other words, the strategy can be used to design network topologies resilient to the isolation of faulty sensors, the removal of failed vehicles, or the failure of a certain number of communication links without losing observability. 
	
	The proposed distributed estimator requires synchronous update at the CAVs with no time-delay, while every CAV needs to know the structure of the network topology in its neighborhood. In case of communication delays over the network of CAVs, the dynamics \eqref{eq_p}-\eqref{eq_m} needs to be updated by considering augmented consensus protocols as in \cite{6571230,DOOSTMOHAMMADIAN2025106260}. The cases of asynchronous networking and delay-tolerant design are our directions of future research. This paper considers linear system model and extension to nonlinear case is left for future research.
	Exploring the integration of distributed observers with advanced local control strategies, such as cooperative adaptive cruise control and automated lane-keeping systems, is another direction of future research.

	\section*{Declarations}
	This work has been supported by the Center for International Scientific Studies \& Collaborations (CISSC), Ministry of Science, Research and Technology
	of Iran.
	
	\bibliographystyle{spmpsci} 
	\bibliography{bibliography}
	
\end{document}